\documentclass[11pt, letterpaper]{article}
\pdfoutput=1

\usepackage[margin=1in]{geometry}
\usepackage{style}
\usepackage{mathtools}
\usepackage{shortcuts}
\usepackage[capitalize]{cleveref}
\usepackage{enumitem}
\usepackage{booktabs}
\usepackage{makecell}
\usepackage{subcaption}
\usepackage{tikz}
\usepackage{pgfplots}
\pgfplotsset{compat=1.10}
\usepgfplotslibrary{fillbetween}
\usetikzlibrary{backgrounds,calc,patterns}

\def\IN{m_{\mathrm{in}}}
\def\OUT{m_{\mathrm{out}}}
\def\eps{\varepsilon}
\newcommand\MM{\mathrm{MM}}
\newcommand\AETriangle{\textsc{AE-Triangle}}
\newcommand\PSAETriangle{\textsc{PS-AE-Triangle}}
\newcommand\tOh{\widetilde{O}}

\crefname{fact}{Fact}{Facts}
\Crefname{fact}{Fact}{Facts}
\crefname{hypothesis}{Hypothesis}{Hypotheses}
\Crefname{hypothesis}{Hypothesis}{Hypotheses}
\crefname{figure}{Figure}{Figures}
\Crefname{figure}{Figure}{Figures}

\title{The Time Complexity of Fully Sparse Matrix Multiplication}
\author{%
	\parbox[t]{.38\textwidth}{\centering%
		Amir Abboud%
		\footnote{This work is part of the project CONJEXITY that has received funding from the European Research Council (ERC) under the European Union's Horizon Europe research and innovation programme (grant agreement No.~101078482). Supported by an Alon scholarship and a research grant from the Center for New Scientists at the Weizmann Institute of Science.}
		\\[.2ex]\small Weizmann Institute of Science}
	\and\parbox[t]{.38\textwidth}{\centering%
		Karl Bringmann%
		\footnote{This work is part of the project TIPEA that has received funding from the European Research Council (ERC) under the European Unions Horizon 2020 research and innovation programme (grant agreement No. 850979).}
		\\[.2ex]\small Saarland University
		\\\small Max Planck Institute for Informatics
		\\[-.6ex]\mbox{}}
	\and\parbox[t]{.38\textwidth}{\centering
		Nick Fischer%
		\footnote{Supported by the project CONJEXITY as above.}
		\\[.2ex]\small Weizmann Institute of Science}
	\and\parbox[t]{.38\textwidth}{\centering%
		Marvin Künnemann%
		\footnote{Research partially supported by the Deutsche Forschungsgemeinschaft (DFG, German Research Foundation) – 462679611.}
		\\[.2ex]\small RPTU Kaiserslautern-Landau}}
\date{}

\begin{document}
\maketitle

\begin{abstract}
\noindent 
What is the time complexity of matrix multiplication of sparse integer matrices with $\IN$ nonzeros in the input and $\OUT$ nonzeros in the output? This paper provides improved upper bounds for this question for almost any choice of~$\IN$ vs.~$\OUT$, and provides evidence that these new bounds might be optimal up to further progress on fast matrix multiplication.

Our main contribution is a new algorithm that reduces sparse matrix multiplication to dense (but smaller) rectangular matrix multiplication. Our running time thus depends on the optimal exponent $\omega(a,b,c)$ of multiplying \emph{dense}~\makebox{$n^a \times n^b$} by $n^b \times n^c$ matrices. We discover that when $\OUT=\Theta(\IN^r)$ the time complexity of sparse matrix multiplication is $\Order(\IN^{\sigma +\epsilon})$, for all $\epsilon > 0$, where $\sigma$ is the solution to the equation~\makebox{$\omega(\sigma-1,2-\sigma,1+r-\sigma) = \sigma$}. No matter what $\omega(\cdot,\cdot,\cdot)$ turns out to be, and for all $r\in(0,2)$, the new bound beats the state of the art, and we provide evidence that it is optimal based on the complexity of the all-edge triangle problem.
  
In particular, in terms of the input plus output size $m = \IN + \OUT$ our algorithm runs in time $O(m^{1.3459})$.
Even for Boolean matrices, this improves over the previous \smash{$m^{\frac{2\omega}{\omega+1}+\epsilon}=O(m^{1.4071})$} bound [Amossen, Pagh; 2009], which was a natural barrier since it coincides with the longstanding bound of all-edge triangle in sparse graphs [Alon, Yuster, Zwick; 1994]. We find it interesting that matrix multiplication can be solved faster than triangle detection in this natural setting. In fact, we establish an equivalence to a \emph{special case} of the all-edge triangle problem.
\end{abstract}

\thispagestyle{empty}
\newpage
\setcounter{page}{1}

% !TEX root = ../paper.tex
\section{Introduction}
Matrix multiplication is one of the most fundamental and important computational problems. Countless papers are devoted it (e.g., the surveys~\cite{Blaser13,GaoJCHWLW23}), including celebrated algorithms that bound the famous exponent $2\leq \omega < 2.3719$ of its time complexity on worst-case $n\times n$ matrices~\cite{Strassen69,Pan78,BiniCRL79,Pan80,Schonhage81,Romani82,CoppersmithW82,Strassen86,CoppersmithW90,CohnU03,CohnKSU05,Stothers10,Williams12,CohnU13,AlmanW21,DuanWZ23} as well as some lower bounds~\cite{Blaser99,Raz03,Shpilka03,Landsberg14,Landsberg18}. Like any other problem, in most applications (theoretical or practical) the input matrices of interest could be \emph{sparse} (or could be sparsified without much harm) in the sense that the number of nonzero entries is $\IN= o(n^2)$. Consequently, the following question has been repeatedly raised by researchers in various domains, seeking a quantitative understanding of the gains from sparsity: 

\begin{center}
	\textit{What is the time complexity of matrix multiplication if the matrices are sparse?}
\end{center}

It is natural to expect a fine-grained answer in the form of a bound that depends on $\IN$ and~$\omega$. However, the answer may not be satisfying if we do not consider the additional parameter $\OUT$ -- the number of nonzeros in the \emph{output} matrix.\footnote{We assume that the matrices are represented as lists of nonzero entries. Logarithmic factors will not matter in our bounds.} This is because even two very sparse matrices with $\IN=O(n)$ could, in pathological cases, produce a very dense output matrix with $\OUT=\Omega(n^2)$, giving a trivial lower bound of $\Omega(\IN^2)$ just to write the output. But this lower bound is too simplistic: in almost all interesting cases $\OUT$ is much closer to $\IN$. For example, two matrices with $O(n)$ nonzeros in \emph{random} locations can be multiplied in $O(n)$ expected time.\footnote{For any $(i,k)$ such that $A[i,k]=1$, in expectation there are only $O(1)$ relevant entries $(k,j)$ such that $B[k,j]=1$.} There is a rich landscape between these two extremes that we cannot capture if we ignore $\OUT$. The dream goal is an ``instance optimal'' algorithm that achieves the best possible time complexity for any input matrices, based on their specific properties. Towards that goal, we are seeking a bound that includes 
\emph{both} $\IN$ \emph{and} $\OUT$.

Considerable effort has gone towards this question, from multiple communities, leading to a state of affairs with several incomparable and complicated bounds. Some of the bounds are discussed below and summarized in Table~\ref{tab:related-work}. In this paper, we clean up the picture by giving (1) a new algorithm for sparse matrix multiplication, (2) an upper bound on its complexity for any setting of~$\IN$ vs.~$\OUT$, and (3) evidence that the achieved bound is tight no matter what the complexity of \emph{dense} (rectangular) matrix multiplication turns out to be.

\subsection{Previous Work}
We briefly review relevant results for sparse matrix multiplication. Let $A,B$ be $n\times n$ matrices, let~$\IN$ denote the number of nonzeros in $A$ and $B$, and~$\OUT$ the number of nonzeros in $A\cdot B$.

The goal of early works was to achieve bounds for \emph{input-sparse} matrix multiplication that get as close as possible to the $O(n^2)$ bound, under minimal assumptions on $\IN$ and $\omega$. The starting point is a folklore approach, first described in~\cite{Gustavson78}, of computing~\smash{$C[i,j] = \sum_{k: A[i,k]\ne 0} A[i,k]\cdot B[k,j]$} over all $i,j$ in time $O(\IN \cdot n)$. An analysis of this approach for uniformly distributed input is given in~\cite{Schorr82}. Using fast rectangular matrix multiplication, a seminal algorithm due to Yuster and Zwick~\cite{YusterZ05} computes $A\cdot B$ faster than dense matrix multiplication whenever~\smash{$\IN = O(n^{\frac{\omega+1}{2}-\epsilon})$}. Their arguments have been adapted to rectangular input matrices in~\cite{KaplanSV06}. Based on recent assumptions from fine-grained complexity, it can be shown that the Yuster-Zwick algorithm is best-possible; see the discussion in \cref{sec:yuster-zwick}.

To obtain further improvements, many subsequent works also exploit sparsity of the output, i.e., take $\OUT$ into account, and thus could hope to achieve subquadratic running times. Randomized algorithms are given in~\cite{Lingas11, Pagh13,GuchtWWZ15,JacobS15, Roche18} and deterministic algorithms in~\cite{AmossenP09,Kutzkov13,GasieniecLLPT17,Kunnemann18,DeepHK20}.  We give a summary in Table~\ref{tab:related-work}. Even if~\makebox{$\IN,\OUT = \Theta(n)$}, only three of the above algorithms can truly beat quadratic running time~$n^{2\pm o(1)}$ (achieved by dense matrix multiplication if $\omega=2$): the Amossen-Pagh bound~\cite{AmossenP09}\footnote{\label{note:ap-bound}The conference version of this paper claimed a better time bound of $O(\IN^{2/3} \OUT^{2/3} + \IN^{0.862} \OUT^{0.408})$ without the restriction $\IN \le \OUT$. On the authors' website this was corrected to a time bound of $O(\IN^{2/3} \OUT^{2/3} + \IN^{0.862} \OUT^{0.546})$ with the restriction $\IN \le \OUT$. (This time bound is the same as what we wrote in Table~\ref{tab:related-work} apart from using updated values for $\alpha,\omega$.) In~\cite{DeepHK20}, the analysis of the Amossen-Pagh algorithm was extended to the case $\OUT \leq \IN$ achieving a running time of \smash{$\IN \cdot (\OUT)^{\frac{2\omega}{\omega+1}-1+\order(1)} = \Order(\IN \cdot \OUT^{0.407})$}.},  and the bound achieved by van Gucht et al.~\cite{GuchtWWZ15} and Roche~\cite{Roche18}\footnote{\label{note:roche-bound}Roche gives a more refined bound, involving an additional parameter $r$. If we only assume that $\IN \ge n$ and use the worst-case choice for $r$, we obtain the stated bound.}. Since our main interest is in subquadratic running times, we will only compare our results to these two bounds.

Some works study additional parameters such as the distribution of nonzeros over the rows and columns~\cite{IwenS09,Roche18}. The closely related setting of error correction for matrix products has been studied in~\cite{GasieniecLLPT17,Roche18}.\footnote{Over rings, this turns out to be essentially equivalent to sparse matrix multiplication, see~\cite{Kunnemann18} and Section~\ref{sec:introduction:sec:apps}.} The communication complexity of output-sensitive matrix multiplication has been studied in~\cite{WilliamsY14}. Finally, output-sensitive quantum algorithms have been given, e.g., in~\cite{BuhrmanS06,VassilevskaWW18,LeGall12}.

% !TEX root = ../paper.tex
\begin{table}[t]
\setlength\extrarowheight{1ex}
\caption{Overview over sparse matrix multiplication algorithms. We list time bounds $t$ with the understanding that there are algorithms with running time $O(t^{1+\epsilon})$ for arbitrarily small $\epsilon > 0$. For the deterministic output-sensitive algorithms in~\cite{Kutzkov13,Kunnemann18}, we assume that a close upper bound on $\OUT$ is given. We highlight in italic the only algorithms that run in time $O(n^{2-\epsilon})$ in the natural setting $\IN,\OUT = \Theta(n)$. To get the numerical bounds, we plug in the current bounds of~\makebox{$\omega \le 2.3719$}~\cite{DuanWZ23} and $\alpha \ge 0.3138$~\cite{LeGallU18}. Related results are given by~\cite{IwenS09, WilliamsY14, GasieniecLLPT17, BuhrmanS06,VassilevskaWW18,LeGall12}.} \label{tab:related-work}
\small
\begin{tabular*}{\textwidth}{p{.33\textwidth}p{.33\textwidth}p{.33\textwidth}}
    \toprule
    Source & Running Time & Notes \\[.5ex]
    \midrule
    Dense matrix multiplication & $n^{\omega}$ & \\
    Gustavson~\cite{Gustavson78} (folklore) & $n\cdot \IN$ &  \\
    Yuster and Zwick~\cite{YusterZ05} & \makecell[lt]{$\IN^{\frac{2(\omega-2)}{\omega-1-\alpha}}n^{\frac{2-\alpha\omega}{\omega-1-\alpha}}+n^{2}$\\$\qquad\le \IN^{0.703}n^{1.187}+n^{2}$}  & only relevant if $\omega > 2$\\ 
	\makecell[lt]{\emph{Amossen and Pagh}{\textsuperscript{\ref{note:ap-bound}}}\\\cite{AmossenP09,DeepHK20}} & \makecell[lt]{$\IN^{2/3} \OUT^{2/3} + \IN^{\frac{(2-\alpha)\omega-2}{(1+\omega)(1-\alpha)}} \OUT^{\frac{2-\alpha\omega}{(1+\omega)(1-\alpha)}}$\\$\qquad\le \IN^{2/3} \OUT^{2/3} + \IN^{0.865} \OUT^{0.543}$} & only if $\IN\le \OUT$, Boolean \\
    Lingas~\cite{Lingas11} & \makecell[lt]{$n^2 \OUT^{\omega/2-1}$\\$\qquad\le n^2\OUT^{0.186}$} & Boolean, randomized\\
    Pagh~\cite{Pagh13} & $\IN + n\cdot \OUT$ & real-valued, randomized\\
    Kutzkov~\cite{Kutzkov13} & $n^2+ n\cdot \OUT^2$ & real-valued\\ 
    Jacob and Stöckel~\cite{JacobS15} & \makecell[lt]{$\IN + n^2(\frac{\OUT}{n})^{\omega-2}$\\$\qquad \le \IN+n^2(\frac{\OUT}{n})^{0.3719}$} & field-valued, randomized\\
    \makecell[lt]{\emph{van Gucht, Williams, Woodruff,}\\\emph{Zhang}~\cite{GuchtWWZ15}} & $\OUT + \sqrt{\OUT}\IN$ & Boolean, randomized \\
    Künnemann~\cite{Kunnemann18} & $\sqrt{\OUT}n^2 + \OUT^2$ & integer-valued\\
	\emph{Roche}{\textsuperscript{\ref{note:roche-bound}}}~\cite{Roche18} & $\OUT + \sqrt{\OUT}\IN$ & field-valued, randomized\\
    \bottomrule
\end{tabular*}
\end{table}

\paragraph{Combinatorial Algorithms}
There has been interest in ``combinatorial'' matrix multiplication algorithms \cite{ArlazarovDKF70,BansalW12,VassilevskaWW18,Chan15,Yu18,DasKS18} that do not exploit algebraic ideas, partly because these algorithms can be more efficient in practice. So far, such methods are only able to save polylogarithmic factors over the naive $O(n^3)$ time complexity in the dense case. In the sparse case, the $\tOh(\OUT + \sqrt{\OUT}\IN)$ bound of \cite{GuchtWWZ15} is combinatorial, and it cannot be improved by more than $n^{o(1)}$ factors without breaking through the cubic bound in the dense case.\footnote{The \emph{Combinatorial Boolean matrix multiplication} conjecture implies that multiplying an $n^a \times n^b$ by an $n^b \times n^a$ matrix requires time $n^{2a+b-\order(1)}$. Since for this problem the input size is trivially bounded by $n^{a+b}$ and the output size is bounded by $n^{2a}$, a combinatorial algorithm for sparse matrix multiplication in time $\Order((\sqrt{\OUT} \IN)^{1-\epsilon})$ would imply a combinatorial algorithm for the aforementioned problem in time \smash{$\Order(n^{2a+b-\epsilon'})$}; a contradiction.}
The goal of this paper is to quantify how well non-combinatorial methods perform in the sparse case.

\paragraph{Rectangular Matrix Multiplication and the \boldmath$\omega(a,b,c)$ Notation}
Our algorithms reduce sparse matrix multiplication to dense (but smaller) \emph{rectangular} matrix multiplication. The running times thus depend on the optimal exponent $\omega(a,b,c)$ of multiplying \emph{dense} $n^a \times n^b$ by $n^b \times n^c$ matrices, for certain values of $a,b,c \geq 0$ depending on the setting. The square case of $\omega=\omega(1, 1, 1)$ implies certain ``naive'' bounds on $\omega(a, b, c)$ in a black-box way by partitioning rectangles into squares. Much better bounds, however, can be obtained in a non-black-box application of the techniques. The most important constant related to rectangular matrix multiplication is $\alpha\leq 1$ defined as the largest constant such that $\omega(1,\alpha,1)=2$. Note that~$\alpha=1$ if and only if $\omega=2$. It has been known for a while that $\alpha > 0$~\cite{Coppersmith82,Coppersmith97,HuangP98,LeGall12}, and the current bound is $\alpha \geq 0.3138$~\cite{LeGallU18}. Other interesting bounds on $\omega(a, b, c)$ where $a, c$ are not $1$ have also been obtained~\cite{LeGall12,LeGallU18}. In a sense, the goal of this work is to settle the sparse setting up to further progress on the rectangular setting, thus reducing the number of research topics by one.

\subsection{Our Results: The Fully Sparse Setting}
In \cref{sec:introduction:sec:general} we present our general bound for any $\OUT=\IN^r$ in the entire feasible range $r \in [0,2]$, and evidence that it is optimal. To highlight our contributions as clearly as possible, we focus in this section on the complexity in terms of the single parameter $m:=\IN+\OUT \approx \max(\IN,\OUT)$ that represents the size of the input plus the output. We call this the \emph{fully sparse} setting because~$m$ bounds the sparsity of both input and output. A high-degree/low-degree algorithm achieves time~\makebox{$O(m^{3/2})$} and the trivial lower bound is $\Omega(m)$; we would like to know the precise exponent. Note that this is a special case of our general result where $r=1$ because in the worst-case~\makebox{$\IN,\OUT=\Theta(m)$}.

The fully sparse setting is not only elegant from a theoretical point of view (because there is only one parameter) but it also received special attention due to natural applications for joins in databases and for transitive closure computation in graphs. See Section~\ref{sec:introduction:sec:apps} for more details. Let us also remark that for the extensively-studied \emph{sparse convolution} problem~\cite{ColeH02,ArnoldR15,ChanL15,Nakos20,BringmannFN21,BringmannFN22}, which is a one-dimensional analogue of our problem, the fully sparse setting is the predominantly studied sparse setting.

\paragraph{Boolean Matrix Multiplication: Beating the Triangle Bound}
Previous work on sparse matrix multiplication has beaten the $m^{3/2}$ bound in the important special case of Boolean matrices over $\{0,1\}$; the state of the art (before our paper) was time \smash{$m^{\frac{2\omega}{\omega+1}+\epsilon} \le O(m^{1.4071})$}, which becomes~$m^{4/3+\epsilon}$ if $\omega=2$~\cite{AmossenP09}. This is a natural barrier because it coincides with the longstanding upper bound for detecting a triangle in a graph on~$m$ edges~\cite{AlonYZ97}.

Triangle detection and (Boolean) matrix multiplication are intimately connected by fine-grained reductions~\cite{VassilevskaWW18} in the dense case. Also in the sparse case, the triangle detection problem can be described as a natural \emph{subset} matrix multiplication problem~\cite{GoldsteinKLP16}: Given two matrices~$A,B$ with $m$ nonzeros, and given a subset~\makebox{$S \subseteq [n] \times [n]$} of $m$ entries, determine if any of the entries in~$(AB)[i,j]$ with $(i,j) \in S$ are nonzero. 

A priori, it is not clear if fully sparse matrix multiplication should be easier or harder than subset matrix multiplication. On the one hand, it feels harder because we not only need to compute $m$ entries in $A B$, we also don't know their locations. But on the other hand, we are promised that~$A B$ only has $m$ nonzeros whereas in the subset problem~$A B$ could be dense.

Our first main result is an algorithm for fully sparse matrix multiplication that beats the triangle bound, improving the complexity from $O(m^{1.4071})$ to $\Order(m^{1.3459})$. To our knowledge, this is the first instance of a natural setting where matrix multiplication is \emph{faster than triangle detection}. To state the precise running time, we introduce another constant related to rectangular matrix multiplication: Let $\mu$ be the (unique) solution to the equation $\omega(\mu, 1, 1) = 2\mu + 1$. The currently best bounds on $\mu$ are $\frac12 \leq \mu \leq 0.5286$~\cite{LeGallU18}, and if $\omega = 2$ then $\mu = \frac12$. This constant naturally appears in several algorithms, including various settings of the All-Pairs Shortest Paths (APSP) problem~\cite{Zwick99,Zwick02,AlonY07,ChanWX21}---most notably, Zwick's algorithm for directed, unweighted APSP in time $O(n^{2+\mu+\epsilon})$---as well as algorithms for dynamic transitive closure~\cite{DemetrescuI05,Sankowski04}.

\begin{theorem}[Fully Sparse Boolean Matrix Multiplication]
\label{thm:det-sparseMM}
Sparse Boolean matrix multiplication is in deterministic time \smash{$\Order(m^{1+\frac{\mu}{1+\mu}+\epsilon}) \leq \Order(m^{1.3459})$}, for all $\epsilon > 0$.
\end{theorem}

Our bound is strictly better than the triangle bound as long as $\omega>2$; otherwise, both bounds become~$m^{4/3+\epsilon}$ which is likely to be the right complexity ``at the end of days''. Interestingly, our algorithm does not need~\makebox{$\omega=2$} (or equivalently, $\alpha = 1$ where $\alpha$ is the aforementioned rectangular matrix multiplication exponent) to achieve exponent $4/3$ but already achieves it under the milder condition that~\makebox{$\mu=\frac12$} (or equivalently, $\alpha \geq \frac12$). In particular, while there are strong barrier results for all currently known approaches to square matrix multiplication (such as the popular \emph{laser method} and its generalizations, applied to the Coppersmith-Winograd tensor) that rule out proving~\makebox{$\omega < 2.3$}~\cite{AmbainisFG15,BlasiakCCGU16,BlasiakCCGU17,AlmanW18a,AlmanW18b,Alman21,ChristandlVZ21}, the best-known barriers for rectangular matrix multiplication only rule out proving $\alpha > 0.625$~\cite{AlmanW18a,ChristandlGLZ20}. Thus, it is conceivable that further progress on fast rectangular matrix multiplication, using the only current set of techniques, leads to sparse Boolean matrix multiplication in time $m^{4/3}$ via our \cref{thm:det-sparseMM}. 

\paragraph{Equivalence}
Can our algorithm be improved? That is, can we either (1) get closer to $m^{4/3}$ with current $\omega(\cdot,\cdot,\cdot)$ or (2) break the $m^{4/3}$ barrier in the future, e.g. if $\omega=2$?
Our next result shows that both (1) and (2) are \emph{equivalent} to making progress on a special case of the \emph{all-edge triangle} problem.

\medskip
Multiple fine-grained complexity conjectures are about triangles in graphs. Underlying all of them is the fundamental statement that \emph{``the only two algorithms for triangle detection are either two-path enumeration or fast matrix multiplication''}. The specific formalization of each hypothesis depends on the family of graphs and the bound achieved when combining these two algorithms. For example, the \emph{Strong Triangle} conjecture~\cite{AbboudV14} states that \smash{$m^{\frac{2\omega}{\omega+1}-o(1)}$} time is required on $m$ edge graphs, and the \emph{Unbalanced Triangle} hypothesis~\cite{KopelowitzV20} states that in a graph on three unbalanced parts with $m_1\le m_2\le m_3$ edges between them the time complexity is \smash{$(m_1m_2)^{1/3}m_3^{2/3-o(1)}$} (even if $\omega=2$). \cref{hypo:AET} in \cref{sec:introduction:sec:general} considers a bound in terms of nodes and edges in some natural regime.\footnote{It is possible to unify all of them under a general and formal hypothesis that address all possible choices of the six parameters for the numbers of nodes and edges between the three parts, but the bound becomes too cumbersome to state.}

In the \emph{all-edge} version, defined next, we must answer for each edge whether it is in a triangle. This version is appealing because (unlike mere detection) we can often prove its hardness based on other famous conjectures. In particular, reductions from 3-SUM \cite{Patrascu10,KopelowitzPP16} and APSP \cite{VassilevskaWilliamsX20} establish an $m^{4/3-{o(1)}}$ lower bound on $m$ edge graphs and serve as a stepping stone for many other reductions~\cite{AbboudBKZ22,AbboudBF23,JinX23}.\footnote{The $3$-SUM conjecture states that no $O(n^{2-\eps})$ time algorithm can decide, given a set of $n$ integers, if there are three that sum to zero. The APSP conjecture states that no $O(n^{3-\eps})$ time algorithm can compute all pair-wise distances in an edge-weighted graph on $n$ nodes.} In this paper, we need an \emph{unbalanced} all-edge triangle version.

\begin{restatable}[$\AETriangle$]{definition}{defaetriangle} \label{def:AET}
The $\AETriangle(x, y, z, m)$ problem is to decide, in a given tripartite graph $G = (X, Y, Z, E)$ with $|X| \leq x,\, |Y| \leq y,\, |Z| \leq z$ and $|E| \leq m$, for each edge $(i, j) \in (X \times Z) \cap E$ whether it is part of a triangle in $G$.
\end{restatable}

It is known that the $3$-SUM and APSP conjectures imply a matching $n^{2-o(1)}$ lower bound for $\AETriangle(n, n, n, n^{3/2})$~\cite{Patrascu10,KopelowitzPP16,VassilevskaWilliamsX20}. If we restrict the size of one of the three parts in the tripartite graph to $\sqrt n$ we get the $\AETriangle(\sqrt n, n, n, n^{3/2})$ problem. Our next theorem shows that (when~\makebox{$\mu=1/2$}) subquadratic time for this problem is possible if and only if the $m^{4/3}$ bound for fully-sparse Boolean matrix multiplication can be broken. It is interesting to note that in the $\AETriangle(\sqrt n, n, n, n^{3/2})$ setting only one of the edge-parts is sparse while the other two are dense. 

\begin{restatable}[Equivalence with \AETriangle]{theorem}{thmequivalenceboolean} \label{thm:equivalence-boolean}
The following two statements are equivalent in terms of deterministic and randomized algorithms:
\begin{enumerate}[label=(\arabic*)]
\item There is some $\epsilon > 0$ such that sparse Boolean matrix multiplication is in time $\Order(m^{1+\frac{\mu}{1+\mu}-\epsilon})$.
\item There is some $\epsilon' > 0$ such that $\AETriangle(n^\mu, n, n, n^{1+\mu})$ is in time $\Order(n^{1+2\mu-\epsilon'})$.
\end{enumerate}
\end{restatable}

Viewed as a hardness result, Theorem~\ref{thm:equivalence-boolean} gives a tight lower bound for sparse matrix multiplication under a natural hypothesis (discussed below) about \AETriangle. In the other direction, it shows that to break the $m^{4/3}$ barrier for all sparse settings, it is enough to break it in a very specific case where one input matrix is rectangular and dense, one input matrix is square and sparse, and the output matrix is rectangular and dense.

\paragraph{Integer Matrix Multiplication (Or over any Ring)} 
In the matrix multiplication literature, algorithms for the Boolean case tend to extend seamlessly to handle entries in $\Int$ (or any ring, with arithmetic operations computable in constant time). In the sparse setting, however, prior work~\cite{AmossenP09} that beat $m^{3/2}$ faces a fundamental challenge when negative numbers are allowed: Many pairs of nonzero entries~$A[i,k]$ and~$B[k,j]$ (with the same $k$) can cancel out, forcing the algorithm to spend too much time without ``gaining'' any nonzero in the output. Notably, such massive cancellations do not only occur in pathological cases, but are actually inherent to interesting applications of sparse \emph{integer} matrix multiplication; one example to error correction is discussed in \cref{sec:introduction:sec:apps}.

An exciting feature of our new algorithm is that it can handle integers just as well, at the cost of introducing randomness, giving a dramatic improvement from time $\tOh(m^{1.5})$~\cite{Roche18} to time~$\Order(m^{1.3459})$. In fact, our algorithm extends to \emph{any} ring $R$. In this case our running time naturally depends on the complexity of dense matrix multiplication over that ring $R$; i.e., the constants $\omega$ and $\mu$ depend on $R$ (see the discussion in \cref{sec:preliminaries}).

\begin{theorem}[Fully Sparse Matrix Multiplication]
\label{thm:integer}
Let $R$ be a ring. Sparse matrix multiplication over $R$ is in randomized time \smash{$\Order(m^{1+\frac{\mu}{1+\mu}+\epsilon})$}, for all $\epsilon > 0$ (assuming an oracle for arithmetic operations over $R$). For $R = \Int$, this running time becomes $\Order(m^{1.3459})$.
\end{theorem}

Since the integer case is only harder than the Boolean case, the bound in Theorem~\ref{thm:integer} is tight unless we break the aforementioned bound of all-edge triangle. In Section~\ref{sec:lower-bounds} we strengthen the equivalence of Theorem~\ref{thm:equivalence-boolean} to show that the integer case is equivalent to a \emph{counting} version of all-edge triangle (Theorem~\ref{thm:equivalence-integer}).

\subsection{The General Bound} \label{sec:introduction:sec:general}
Recall that our motivating question was about identifying the optimal complexity for any choice of $\IN$ vs. $\OUT$, i.e., $\OUT = \IN^r$ for any $r\in [0,2]$.\footnote{The two extreme examples that show that $0 \leq r \leq 2$ are the following. If the only nonzero entries are $A[1,k]=1$ and $B[k,1]=1$ for all $k\in[n]$ then $\IN=n$ while $\OUT=1$. If the only nonzero entries are $A[i,1]=1$ and $B[1,j]$ for all $i,j \in [n]$ then $\IN=n$ but $\OUT=n^2$.} In practice, it would be ideal if there is a single algorithm that is guaranteed to achieve the optimal complexity when run on matrices with $\IN$ nonzeros, no matter what $\OUT$ turns out to be. With a more intricate analysis and under similar assumptions, our new algorithm accomplishes this.

We suggest that the time complexity of sparse matrix multiplication is $\IN^{\sigma(r)}$ when $\OUT=\IN^r$, for the following definition of the exponent $\sigma(r) \in [1,2]$ that \emph{only depends} on the optimal exponent of \emph{dense} rectangular matrix multiplication. (It is not clear that $\sigma(r)$ is well-defined but we prove it in \cref{sec:input-sparse:sec:sparse-exp}.)

\begin{restatable}[Exponent of Sparse Matrix Multiplication]{definition}{defsparseexp} \label{def:sparse-exp}
Let $r \in [0, 2]$. We define $\sigma(r)$ as the unique solution $\sigma$ to the equation $\omega(\sigma - 1, 2 - \sigma, 1 + r - \sigma) = \sigma$.
\end{restatable}

The following are our two main theorems. They generalize \cref{thm:det-sparseMM,thm:integer} from $r=1$ to any $r \in [0,2]$; the first gives a deterministic algorithm for the Boolean case and the second gives a randomized algorithm for integers (or any ring).

\begin{restatable}[Deterministic Sparse Boolean Matrix Multiplication]{theorem}{thmsparsemmboolean} \label{thm:sparse-mm-boolean}
Sparse Boolean matrix multiplication with input size $\IN$ and output size $\OUT = \IN^r$ is in deterministic time \smash{$\Order(\IN^{\sigma(r)+\epsilon})$}, for all $\epsilon > 0, r \in [0,2]$.
\end{restatable}

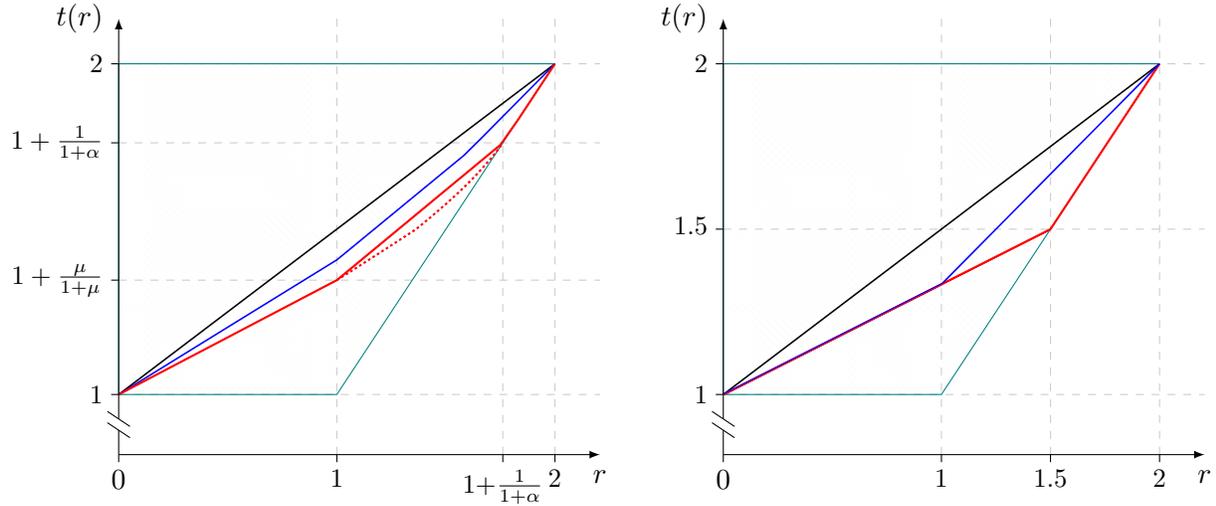
\cref{fig:sparse-exp-plots} plots our new bound against the previous results highlighted in Table~\ref{tab:related-work} (in the Boolean case) both for (a) the current bounds on $\sigma(r)$ where we beat the state of the art for all $r$, and (b) for $\omega=2$ in which case we get an improvement for all $r>1$ (note that this is the more natural setting). Let us say a few words about how we bound $\sigma(r)$ (all the details are in Section~\ref{sec:input-sparse:sec:sparse-exp}). If~\makebox{$\omega=2$} then it is easy to calculate $\omega(a,b,c)$ for any $a,b,c$ and we get that $\sigma(r) = \max\set{1 + \frac r3, r}$. The situation is more complicated when using the current bounds on $\omega(a,b,c)$~\cite{LeGallU18,DuanWZ23}. First we prove that $\sigma(r)$ is a convex function of $r$. Then we apply either of two methods: one algebraic and the other numeric. In the first one, we evaluate $\sigma(r)$ at certain strategic values of $r$, and then interpolate between these points; the bound we obtain is
\begin{equation} \label{eq:gkangfnas}
	\sigma(r) \leq \max\set*{1 + r \cdot \frac{\mu}{1+\mu}, \frac{(2+\alpha)\mu}{1+\mu} + r \cdot \frac{1-\alpha\mu}{1+\mu}, r},
\end{equation}
and it is depicted by the thick red line in Figure~\ref{fig:sparse-exp-plots}. In the second method, we feed all known bounds on~$\omega(a,b,c)$~\cite{LeGallU18} into a linear program that finds the best bound on $\sigma(r)$; this bound is the dotted red line in \cref{fig:sparse-exp-plots}. In any case, further improvements to our bound within the shaded area are not unconditionally impossible, but would refute the natural hypothesis discussed below.

% !TEX root = ../paper.tex
\begin{figure}[t]
\caption{Plots of the exponent $t(r)$ in the running time $O(\IN^{t(r)+\varepsilon})$, for any $\varepsilon>0$, of sparse Boolean matrix multiplication where $\OUT = \IN^r$. The shaded teal region depicts the relevant area between the trivial $\Omega(\IN + \OUT)$ lower bound and $O(\IN^2)$ upper bound. The black lines depict the  $\Order(\IN \sqrt{\OUT})$ upper bound~\cite{GuchtWWZ15,Roche18}. The blue lines depict the upper bound of \cite{AmossenP09} as analyzed by~\cite{DeepHK20}. The red lines depict our bounds on $\sigma(r)$; the thick lines correspond to~(\ref{eq:gkangfnas}) and the dotted line was obtained via numerical computations. All lines except the blue hold for randomized integer matrix multiplication as well.} \label{fig:sparse-exp-plots}%
\newcommand\plotsparseexp[3]{%
    \begin{tikzpicture}[%
        x=2.9cm, y=4.4cm,
        >=latex,
        every node/.append style={inner sep=0pt},
        fill between/on layer={background},
        baseline={(0, 1)},
    ]
        % lengths
        \def\offsetlen{.8cm}
        \def\offsetsep{.15cm}
        \def\overshoot{.6cm}

        % grid lines
        \foreach\x/\l in {#2}{
            \draw[dashed, black!20] (\x, 1) ++(0, -\offsetlen) -- (\x, 2) -- ++(0, \overshoot);
            \draw (\x, 1) ++(0, -\offsetlen) -- ++(0, -.1cm) node[below=.1cm] {\small $\l$};
        }
        \foreach\y/\l in {#3}{
            \draw[dashed, black!20] (0, \y) -- (2, \y) -- ++(\overshoot, 0);
            \draw (0, \y) -- ++(-.1cm, 0) node[left=.1cm] {\small $\l$};
        }

        % plots
        \def\OMEGA{2.3727}
        \def\RHO{0.5286}
        \def\ALPHA{0.3139}
        #1

        % axes
        \draw[->] (0, 1) ++(0, -\offsetlen) -- ++(2, 0) -- ++(\overshoot, 0) node[below=.2cm] {$r$};
        \draw (0, 1) -- ++(0, -.5 * \offsetlen + 0.5 * \offsetsep);
        \draw (0, 1) ++(0, -\offsetlen) -- ++(0, .5 * \offsetlen - 0.5 * \offsetsep);
        \draw (0, 1) ++(-.15cm, -.5 * \offsetlen + .1cm + 0.5 * \offsetsep) -- ++(.3cm, -.2cm);
        \draw (0, 1) ++(-.15cm, -.5 * \offsetlen + .1cm - 0.5 * \offsetsep) -- ++(.3cm, -.2cm);
        \draw[->] (0, 1) -- (0, 2) -- ++(0, \overshoot) node[left=.2cm] {$t(r)$};

        % axis zero label
        \draw (0, 1) ++(0, -\offsetlen) -- ++(0, -.1cm) node[below=.1cm] {$0$};
    \end{tikzpicture}
}%
\vspace{-.25cm}
\rule{\textwidth}{\arrayrulewidth}
\vspace{-.1cm}

\begin{minipage}{\textwidth}
\begin{subfigure}[t]{.5\textwidth}
\centering
\plotsparseexp{%
    \draw[domain=0:2, variable=\c, teal, name path=UB] (0, 1) -- (0, 2) -- (2, 2);
    \draw[domain=0:2, variable=\c, teal, name path=LB] plot ({\c}, {max(1, \c)});
    \tikzfillbetween[of=UB and LB]{teal, opacity=0.12, pattern=north west lines, pattern color=teal};
    \draw[domain=0:2, variable=\c, semithick, black] plot ({\c}, {1+\c/2});
    \draw[domain=0:2, variable=\c, semithick, blue] plot({\c}, {max(1+\c*(2*\OMEGA/(\OMEGA+1)-1), ((2-\ALPHA)*\OMEGA-2)/((1+\OMEGA)*(1-\ALPHA))+\c*(2-\ALPHA*\OMEGA)/((1+\OMEGA)*(1-\ALPHA)), 2/3+\c*2/3)});
    \draw[domain=0:2, variable=\c, thick, red] plot({\c}, {max(1+\c*\RHO/(1+\RHO), (2+\ALPHA)*\RHO/(1+\RHO)+\c*(1-\ALPHA*\RHO)/(1+\RHO), \c)});
    \draw[domain=1:1.7, variable=\c, thick, red, dashed, dash pattern=on 1pt off 1pt] plot file {figures/sparse-exp-plots.dat};
}{1/1, 1.7610/1{+}\frac{1}{1+\alpha}, 2/2}{1/1, 1.3458/1+\frac{\mu}{1+\mu}, 1.7610/1+\frac{1}{1+\alpha}, 2/2}
\vspace{-1.5cm}
\caption{The plots for the current bounds on (rectangular) matrix multiplication~\cite{LeGallU18,DuanWZ23}.}
\end{subfigure}
\begin{subfigure}[t]{.5\textwidth}
\centering
\plotsparseexp{%
    \draw[domain=0:2, variable=\c, teal, name path=UB] (0, 1) -- (0, 2) -- (2, 2);
    \draw[domain=0:2, variable=\c, teal, name path=LB] plot ({\c}, {max(1, \c)});
    \tikzfillbetween[of=UB and LB]{teal, opacity=0.12, pattern=north west lines, pattern color=teal};

    \draw[domain=0:2, variable=\c, semithick, black] plot ({\c}, {1+\c/2});
    \draw[domain=0:2, variable=\c, thick, red] plot({\c}, {1+max(\c/3, \c-1)});
    \draw[domain=0:1, variable=\c, thick, blue] plot({\c}, {max(2/3+\c*2/3, 1 + \c/3)});
    \draw[domain=1:2, variable=\c, semithick, blue] plot({\c}, {max(2/3+\c*2/3, 1 + \c/3)});
    \begin{scope}
        \clip[overlay]
          ($(0, 1)!-2pt!(1.5, 1.5)$) coordinate (tmpA)
          -- ($(1.5, 1.5)!-2pt!(0, 1)$) coordinate (tmpB)
          -- ($(tmpB)!2pt!90:(tmpA)$)
          -- ($(tmpA)!2pt!-90:(tmpB)$)
          -- cycle
        ;
        \draw[red, thick] (0, 1) -- (1.5, 1.5);
    \end{scope}
}{1/1, 1.5/1.5, 2/2}{1/1, 1.5/1.5\vphantom{\frac{1}{1+\alpha}}, 2/2}
\vspace{-1.5cm}
\caption{The plots if $\omega = 2$.}
\end{subfigure}
\end{minipage}

\vspace{-.5cm}
\rule{\textwidth}{\arrayrulewidth}
\end{figure}

The next theorem proves that our same bound holds for randomized algorithms in the integer case despite the cancellations; note that in this case the blue line in Figure~\ref{fig:sparse-exp-plots} does not apply and our improvements are bigger.

\begin{restatable}[Sparse Matrix Multiplication]{theorem}{thmsparsemminteger} \label{thm:integer-bivariate}
Let $R$ be a ring. Sparse matrix multiplication over~$R$ with input size $\IN$ and output size $\OUT = \IN^r$ is in randomized time \smash{$\Order(\IN^{\sigma(r)+\epsilon})$}, for all~\makebox{$\epsilon > 0, r \in [0,2]$}.
\end{restatable}

In \cref{sec:introduction:sec:apps} we mention one application of \cref{thm:integer-bivariate} that does not follow from the results on the fully sparse setting ($r=1$), namely to error correction. Interestingly, based on this connection, we discuss that derandomizing our result (for integers) is at least as hard as derandomizing Freivalds' classical algorithm \cite{Freivalds79}.

Finally, we will isolate a concrete version of the all-edge triangle problem that must be cracked before our sparse matrix multiplication bound can be improved \emph{even for a single $r$}. We call the following unbalanced version $\PSAETriangle$ (``partially sparse'' $\AETriangle$) because only one edge-part is sparse while the other two are assumed to be dense (since the restriction on the number of edges only applies to one edge-part).

\begin{restatable}[$\PSAETriangle$]{definition}{defpsaetriangle} \label{def:ps-ae-triangle}
The $\PSAETriangle(x, y, z, m)$ problem is to decide, in a given tripartite graph $G = (X, Y, Z, E)$ with $|X| \leq x,\, |Y| \leq y,\, |Z| \leq z$ and $|E \cap (Y \times Z)| \leq m$, for each edge $(i, j) \in (X \times Z) \cap E$ whether it is part of a triangle in $G$.
\end{restatable}

Recall from the discussion above \cref{def:AET} that the fundamental underlying hypothesis is that the only two algorithms for triangle detection are either two-path enumeration or fast matrix multiplication. For any choice of parameters, it is easy to compute the bound achieved by combining these two algorithms, giving rise to the following hypothesis. 

\begin{restatable}[$\PSAETriangle$]{hypothesis}{hypoaet} \label{hypo:AET}
For all $a, b, c \geq 0$, the $\PSAETriangle(m^a, m^b, m^c, m)$ problem cannot be solved in time $\Order(m^{\min\set{1+a,\, \omega(a, b, c)}-\epsilon})$, for any $\epsilon > 0$.
\end{restatable}

The exponent in \cref{def:sparse-exp} is \emph{the} sparse matrix multiplication exponent unless \cref{hypo:AET} fails. 

\begin{restatable}[Hardness under \PSAETriangle]{theorem}{thmlowerboundpsaetriangle} \label{thm:lower-bound-ps-ae-triangle}
Let $r \in [0, 2]$. For any $\epsilon > 0$, sparse Boolean matrix multiplication with input size $\IN$ and output size $\OUT = \IN^r$ cannot be solved in time \smash{$\Order(\IN^{\sigma(r) - \epsilon})$}, unless the $\PSAETriangle$ hypothesis fails.
\end{restatable}

Interestingly, we do not even need this assumption for the full range of $r$ to give evidence of optimality of our algorithm. 
Indeed, our algorithm runs in almost-linear time $O(\OUT^{1+\epsilon})$ for arbitrary~\makebox{$\epsilon > 0$}---and thus is \emph{unconditionally almost optimal}---when \smash{$\OUT = \Omega(\IN^{1.762}) = \Omega((\IN)^{1+\frac1{1\alpha}})$}. If $\omega=2$, the condition simplifies to $\OUT \ge \Omega(\IN^{1.5})$.

\subsection{Technical Overview}
To ease the readability of this paper, we give a concise technical overview of our results. Our algorithms draw from a colorful set of techniques, ranging from basic combinatorial heavy/light decompositions to ideas from sparse recovery, derandomization, and algebra. While none of these ideas are particularly new to the context of matrix multiplication, we manage to refine and combine them in a novel way, leading to our results.

To obtain our algorithms for sparse matrix multiplication, we split the task into two major steps. These steps separately deal with the output- and input-sparsity of the problem. For simplicity we focus on the fully sparse setting ($\IN + \OUT \leq m$) in this overview.

\paragraph{Step 1: Output Densification}
In the first step, we design a reduction from fully sparse matrix multiplication to \emph{input-sparse} matrix multiplication of $x \times y \times z$ rectangular matrices. That is, after step 1 we can assume that $x z = \Order(m)$, at the small cost of worsening the running time by a lower-order factor.

The basic idea is to \emph{compress} $A$ to a thinner matrix $A'$ such that (1) the product $A' B$ becomes dense, and (2) we can recover $A B$ from $A' B$. A natural approach to this task, inspired by sparse recovery, is to let $g = \frac{xz}{m}$ and to compress the matrix $A$ by randomly grouping columns into groups of size $g$. Let $A'$ be the $(x/g) \times y$ matrix obtained by adding all columns in a group. Note that~\makebox{$C' = A' B$} is exactly the matrix obtained from $C = A B$ by applying the same compression. Since the hashing was random, we expect that many entries $(i, j)$ are \emph{isolated} in the sense that~\makebox{$C[i, j]$} is the only nonzero entry in its group in the $i$-th row. These isolated entries can be efficiently recovered by some more tricks to identify $j$. For integer matrices, for instance, we can identify $j$ by accessing~\makebox{$C[i, j]$} and $j \cdot C[i, j]$ (both of which can be computed by the previous compression) and taking their quotient.

This main idea---to compress the sparse output matrix to a dense rectangular matrix via sparse recovery techniques---was employed in several previous works~\cite{IwenS09,Pagh13,JacobS15,GuchtWWZ15}. See also the detailed treatment in~\cite{Stockel15}, including implementations of this idea in the external memory model. Unfortunately, this setup is hard to generalize to arbitrary rings and hard to derandomize. To obtain the full strength of our results, we instead suggest the following twist:
\begin{enumerate}
	\setlength\parskip{0pt plus 1pt}
	\item First, assume that we know a small superset $S$ of the support $\supp(C) = \set{(i, j) : C[i, j] \neq 0}$ of size at most $|S| \leq 2m$, say. We apply the same grouping as before, and consider an entry~$(i, j) \in S$ as \emph{isolated} if there is no other entry $(i, j') \in S$ with $j$ and $j'$ belonging to the same group. The nice insight is that we can check for all pairs $(i, j) \in S$ whether they are isolated, and for isolated pairs recover the corresponding entries $C[i, j]$ \emph{without the need to identify $j$.} (For the details see \cref{lem:recover}.)
	
	With some additional effort this step can be derandomized: The key idea is that we can choose the groups \emph{one by one} deterministically, following the method of conditional expectations. We remark that this derandomization is necessary to obtain our deterministic algorithm for Boolean matrix multiplication (\cref{thm:det-sparseMM}).

	\item Next, we remove the assumption that $S$ is known in advance. For simplicity in this overview, we assume that the matrices $A$ and $B$ contain nonnegative integers. Then we can compute~$S$ by \emph{recursively} calling our sparse matrix multiplication algorithm: Let $A'$ be the $(x/2) \times y$ matrix obtained from $A$ by merging and adding up pairs of adjacent rows. We compute the matrix product $C' = A' B$ recursively. Then, we select $S$ to be all positions $(i, j)$ that under the pairing could possibly lead to nonzero entries in $C'$. It is easy to check that $S \supseteq \supp(C)$ and that $|S| \leq 2m$. Moreover, with each recursive call we half $x$ and therefore the recursion incurs only a logarithmic factor to the running time. (For the details see \cref{lem:densification-nonnegative}.)
	
	If we allow randomization, this idea indeed generalizes to integers with polylogarithmic overhead (\cref{lem:densification-integer}) and to arbitrary rings with subpolynomial \smash{$2^{\widetilde\Order(\sqrt{\log m})}$} overhead (\cref{lem:densification-rings}).
\end{enumerate}

\paragraph{Step 2: Input-Sparse Matrix Multiplication}
Recall that after step 1, we can assume that~\makebox{$x z = \Order(m)$}. It thus remains to solve an instance of input-sparse rectangular $x \times y \times z$ matrix multiplication (\cref{lem:input-sparse}). In previous works, Yuster and Zwick studied the complexity of input-sparse $x \times x \times x$ matrix multiplication~\cite{YusterZ05}, and Kaplan, Sharir and Verbin studied input-sparse~\makebox{$x \times y \times x$} matrix multiplication~\cite{KaplanSV06}. We emphasize that both of these settings do not suffice in our context, as for many cases we indeed have three distinct parameters~$x, y, z$ (this happens even in the fully sparse case, due to step~1). Nevertheless, we reuse the same simple algorithmic idea behind their algorithms, namely a heavy/light decomposition. The difficulty here does not lie in the algorithm, but in the analysis.

We may assume without loss of generality that $x \leq z$. Let $\Delta$ be a parameter to be determined later. We say that a column $k$ in $A$ is \emph{light} if it contains at most $\Delta$ nonzero entries, and \emph{heavy} otherwise. We split the matrix $A$ into two submatrices $A_{\text{light}}$ and $A_{\text{heavy}}$ consisting of the light and heavy columns, respectively. We similarly split $B$ into $B_{\text{light}}$ and $B_{\text{heavy}}$ where the former matrix consists exactly of all light rows $k$. To compute the product $AB$, it suffices to compute the products~\makebox{$A_{\text{light}} B_{\text{light}}$} and $A_{\text{heavy}} B_{\text{heavy}}$ separately.

We compute the light product $A_{\text{light}} B_{\text{light}}$ in time $\Order(m \Delta)$ by enumerating, for each nonzero entry $B[k, j]$, the at most $\Delta$ relevant entries from $A$. The heavy product $A_{\text{heavy}} B_{\text{heavy}}$, on the other hand, we compute by fast rectangular matrix multiplication. The running time of this algorithm can be bounded by
\begin{equation*}
	\Order\parens*{m \Delta + \max_{\Delta \leq x \leq m/x} \MM(x, \tfrac m\Delta, \tfrac mx)},
\end{equation*}
crucially using that after step 1 we have $x z = \Order(m)$. In~\cite{YusterZ05,KaplanSV06} the authors use coarse bounds on the complexity of rectangular matrix multiplication~\cite{HuangP98} to bound this expression. Our goal is an optimal answer though, so we cannot afford to be lossy in the analysis.

So what is the complexity of \smash{$\max_{\Delta \leq x \leq m/x} \MM(x, \frac m\Delta, \frac mx)$}? The maximum could be attained at the extreme cases $\MM(\Delta, \frac{m}{\Delta}, \frac{m}{\Delta})$ or $\MM(\sqrt m, \frac{m}{\Delta}, \sqrt m)$, or anywhere in between. Unfortunately, these extreme cases seem incomparable and not reducible to each other when we treat dense matrix multiplication as a black-box algorithm (to reduce between the extreme cases, we would have to increase one dimension, but decrease another).

As it does not seem possible to us to give a combinatorial answer to this question, we peek behind the curtain of algebraic matrix multiplication. It turns out that, using light insights on the matrix multiplication tensor (\cref{fact:omega-subadditive}), we can indeed prove that the first extreme case $\MM(\Delta, \frac m\Delta, \frac m\Delta)$ is dominating. Therefore, by setting \smash{$\Delta = m^{\frac{\mu}{1+\mu}}$} where $\mu$ is the solution to the equation~\makebox{$\omega(\mu, 1, 1) = 1 + 2\mu$}, the running time becomes
\begin{equation*}
	\Order(m \cdot m^{\frac{\mu}{1+\mu}} + \MM(m^{\frac{\mu}{1+\mu}}, m^{\frac{1}{1+\mu}}, m^{\frac{1}{1+\mu}})) = \Order(m^{1 + \frac{\mu}{1+\mu}} + m^{\frac{\omega(\mu, 1, 1)}{1+\mu}}) = \Order(m^{1+\frac{\mu}{1+\mu}+\epsilon}),
\end{equation*}
for any $\epsilon > 0$. We remark that this part of the proof works for any ring $R$, too: The light case is purely combinatorial, and in the heavy case we use fast matrix multiplication for that particular ring $R$.

\paragraph{Equivalence with All-Edges Triangle}
Let us also provide some explanation why sparse Boolean matrix multiplication is equivalent to $\AETriangle$ (\cref{thm:equivalence-boolean}). One direction is simple: Boolean matrix multiplication can be viewed as the problem of determining, for each pair of nodes $(i, j) \in X \times Z$ in a tripartite graph~$G = (X, Y, Z, E)$, whether the pair is connected by a 2-path. To solve $\AETriangle$ we thus report all pairs $(i, j)$ that are connected by a 2-path and by an edge (which together form a triangle). For the parameters in \cref{thm:equivalence-boolean} the computed Boolean matrix product is indeed sparse.

The other direction is more interesting; we give some intution: $\AETriangle(\Delta, \frac{m}{\Delta}, \frac{m}{\Delta}, m)$ is essentially equivalent to the hard case in our previously outlined algorithm, when $x = \Delta$. This does \emph{not} render our reduction straightforward though, as some other cases for $\Delta \leq x \leq m/x$ might be equally hard. To make the reduction work, we prove that $x = \Delta$ is indeed the \emph{hardest} case, which requires very strong numeric bounds on some rectangular matrix multiplication exponents (\cref{lem:rho-bounds}), based on the recent work by Le~Gall and Urrutia~\cite{LeGallU18}.

\subsection{Some Applications} \label{sec:introduction:sec:apps}
Our algorithm could be an appealing choice for anyone using matrix multiplication on sparse data. The applications of matrix multiplication are endless, and in each of them one could now claim a better bound if the input and output happen to be sparse. Let us mention only a few concrete examples from TCS.

\paragraph{Join-Project Queries in Relational Databases.}
A natural operation in relational databases are \emph{collapsing join-project queries} of two tables (aka \emph{composition} or \emph{set-intersection joins}), see e.g.~\cite{AmossenP09,GuchtWWZ15,DeepHK20}. Here, we join two tables $R(a,b)$ and $S(b,c)$ on a shared key $b$, followed by projection on $a$ and $c$. As an example, given databases such as DBLP, such queries can be used to determine all pairs of researchers who co-authored a paper together. 

Answering these queries is naturally \emph{equivalent} to sparse Boolean matrix multiplication~\cite{AmossenP09,GuchtWWZ15,DeepHK20}: Here, $\IN$ denotes the size of the two given tables $R$ and $S$, while $\OUT$ denotes the output size of the query. Exploiting this connection, one can answer join-project queries faster than naively evaluating a full join, followed by a projection. Our results conditionally resolve the time complexity of answering these queries.

\paragraph{Transitive Closure in Graphs.}
Consider the problem of computing the transitive closure of a directed graph $G$. Denoting by $m$ the (a priori unknown) number of edges of the transitive closure of~$G$, previous output-sensitive algorithms solve the problem in time~\smash{$\tOh(m^{3/2})$}~\cite{GuchtWWZ15,BorassiCH16}. Using the well-known approach of computing a transitive closure via $O(\log n)$ Boolean matrix products and observing that each of the involved matrices have at most $m$ nonzeros, we obtain the following result as an immediate corollary of Theorem~\ref{thm:det-sparseMM}.

\begin{theorem}[Transitive Closure]
There is a deterministic algorithm computing the transitive closure of a directed graph in time \smash{$\Order(m^{1+\frac{\mu}{1+\mu}+\epsilon}) = \Order(m^{1.3459})$}, for all $\epsilon > 0$, where $m$ is the number of edges in the transitive closure. 
\end{theorem}

This running time transfers to recognition of comparability graphs, see~\cite{BorassiCH16}.

\paragraph{Error Correction}
The fact that we obtain (randomized) algorithms even for rings can be used for error correction in matrix products, by a \emph{deliberate} use of cancellations. In this setting, we are given matrices~$A,B$ and a possibly faulty matrix product $\widetilde{C}\approx AB$. E.g., erroneous entries may have been introduced during transmission of the correctly computed result. The task is to compute the errors $AB-\widetilde{C}$ in order to correct~$\widetilde{C}$ to the true matrix product $AB$. Let $m$ denote the number of nonzeros in the input matrices $A,B,\widetilde{C}$ and let $z$ denote the number of nonzeros in $AB-\widetilde{C}$, i.e., the number of errors. If $z\ll m$, can we correct~$\widetilde{C}$ to the true matrix product faster than computing the product from scratch?

This setting has been studied explicitly e.g.\ in~\cite{GasieniecLLPT17,Kunnemann18, Roche18}. Over rings, error correction can be reduced to sparse matrix multiplication as follows: Given \smash{$A,B,\widetilde{C}$} with $m$ nonzeros, we can construct~$A', B'$ with $O(m)$ nonzeros in time $\tOh(m)$ such that $A'B' = AB-\widetilde{C}$, see~\cite[Proposition 3.1]{Kunnemann18}.\footnote{We sketch the argument here: As discussed in \cref{sec:preliminaries}, we can assume that $A,B,\widetilde{C}$ are $m \times m$ matrices. The desired matrices can be obtained as $A' = ( A \mid -I)$ and $B'=\begin{pmatrix} B \\ \widetilde{C} \end{pmatrix}$. Note that $A'B' = AB-\widetilde{C}$ and that $A',B'$ have $O(m)$ nonzeros. \label{note:errorcorrection}} Thus, correcting $z$ errors over rings reduces to sparse matrix multiplication with~\makebox{$\IN = m$} and $\OUT = z$. Hence, from \cref{thm:integer-bivariate}, we obtain the following corollary.

\begin{theorem}[Sparse Matrix Product Correction]
Consider matrices $A,B,\widetilde{C}$ with $m$ nonzeros, over some ring $R$. If $\widetilde{C}$ differs from $AB$ in at most $O(m^r)$ entries, where $r \in [0, 2]$, then we can compute $AB$ in randomized time $\Order(m^{\sigma(r)+\epsilon})$, for all $\epsilon > 0$. 
\end{theorem}

Roughly speaking, whenever we have fewer errors in $\widetilde{C}$ than nonzeros in $A,B,AB$, we can beat computing~$AB$ from scratch, even for sparse matrices. In particular, if $\omega=2$, we obtain a randomized algorithm correcting up to $z\ge 1$ errors in time $(z + mz^{1/3})^{1+\epsilon}$ for all $\epsilon > 0$.

The above reduction from error correction to sparse matrix multiplication also shows why derandomizing \cref{thm:integer-bivariate} is challenging, as it would derandomize Freivalds' algorithm \emph{even for sparse matrices}: If we could obtain the same running time \emph{deterministically}, we could in particular check, given $A,B,\widetilde{C}$ with $m$ nonzeros, whether $\widetilde{C}=AB$ in almost-linear time $m^{1+\epsilon}$ as follows. Simply start to run the deterministic algorithm on $A',B'$ as constructed above (see \cref{note:errorcorrection}). If $A'B'=0$, i.e., $\widetilde{C}=AB$, it returns an answer within some running time bound $T(m)=m^{1+\epsilon}$ guaranteed for $z=0$ errors. If the number of steps of the algorithm ever exceeds $T(m)$ or the result matrix is different from $0$, we reject, otherwise we accept. This verifies $\widetilde{C}$ in deterministic time $m^{1+\epsilon}$, for all $\epsilon > 0$.

Note that Freivalds' algorithm~\cite{Freivalds79} gives a randomized $\tOh(m)$-time solution, however, a derandomization is still open even for dense matrices, see, e.g.,~\cite{KimbrelS93,NaorN93,Kunnemann18}.

\subsection{Outline}
The rest of this paper is structured as follows. We give some formal preliminaries and some facts on fast matrix multiplication in \cref{sec:preliminaries}. In \cref{sec:densification,sec:input-sparse} we respectively prove the two steps of our algorithm---the output densification and the input-sparse algorithm. In \cref{sec:lower-bounds} we establish lower bounds and equivalences with the All-Edges Triangle problem. In \cref{sec:conclusions} we then conclude with some open problems. Finally, in \cref{sec:yuster-zwick} we provide an interesting simple lower bound for Yuster and Zwick's input-sparse matrix multiplication algorithm.
% !TEX root = ../paper.tex
\section{Preliminaries} \label{sec:preliminaries}
We denote the integers by $\Int$, the nonnegative integers by $\Nat$ and we write $[n] = \set{1, \dots, n}$. We write \smash{$\widetilde\Order(T) = T (\log T)^{\Order(1)}$} to suppress polylogarithmic factors. We say that an event happens \emph{with high probability} if it happens with probability at least $1 - n^{-c}$, where $n$ is the input size of the problem and $c$ is an arbitrarily large constant. Throughout, by \emph{randomized algorithm} we mean a Monte Carlo randomized algorithm that succeeds with high probability.

\paragraph{Dense Matrix Multiplication}
Let $R$ be a ring. We write $\MM_R(x, y, z)$ to denote the minimum number of arithmetic operations necessary to multiply an $x \times y$ matrix with an $y \times z$ matrix over~$R$ (by an arithmetic circuit or an equivalent model~\cite{Blaser13}). The matrix multiplication exponent~$\omega_R(a, b, c)$ is defined as
\begin{equation*}
    \omega_R(a, b, c) = \inf\set{\tau \in \Real : \MM_R(\ceil{n^a}, \ceil{n^b}, \ceil{n^c}) = \Order(n^\tau)}.
\end{equation*}
Most of the time the ring $R$ is clear and we omit the subscript $R$.\footnote{While in principle, $\omega_R(a, b, c)$ can depend on the underlying ring, there is no indication that  two rings actually have different complexities. For fields $F$, it is known that $\omega_F$ can only depend on the characteristic of the field~$F$~\cite{Schonhage81}.} We write $\omega = \omega(1, 1, 1)$ for the exponent of square matrix multiplication; for fields of characteristic $0$ the current state-of-the-art bounds are~\makebox{$2 \leq \omega \leq 2.3719$}~\cite{DuanWZ23}.

Following standard notation, we define two more constants concerned with the complexity of rectangular matrix multiplication: Let $\alpha = \max\set{\alpha \in \Real : \omega(\alpha, 1, 1) \leq 2}$ and let $\mu$ be the (unique) solution to the equation~\makebox{$\omega(\mu, 1, 1) = 1 + 2\mu$}. The best known bounds on $\alpha$ and $\mu$ are $0.31389 \leq \alpha \leq 1$ and~\smash{$\frac12 \leq \mu \leq 0.5286$}~\cite{LeGall12,LeGallU18} (for fields of characteristic $0$).\footnote{It is not necessarily obvious that $\alpha$ and $\mu$ are well-defined. For $\alpha$, we remark that since $\omega$ is a continuous function (\cref{fact:omega-continuous}) the set $\set{\alpha : \omega(\alpha, 1, 1) \leq 2}$ is closed and thus its maximum exists. For $\mu$, we note that both sides of the equation~\makebox{$\omega(\mu, 1, 1) = 1 + 2\mu$} are continuous functions that must meet in some point due to the trivial bounds on $\omega$ (\cref{fact:omega-trivial-bounds}). Using \cref{fact:omega-subadditive}, one can show that the solution is unique.} We routinely rely on the following well-known facts; for proofs, refer for instance to the survey by Bläser~\cite{Blaser13} or classic work such as~\cite{LottiR83}.

\begin{fact} \label{fact:omega-continuous}
$\omega(\cdot, \cdot, \cdot)$ is a continuous function.
\end{fact}

\begin{fact} \label{fact:omega-trivial-bounds}
For any $a, b, c \geq 0$, we have $\max\set{a + b, a + c, b + c} \leq \omega(a, b, c) \leq a + b + c$.
\end{fact}

\begin{fact} \label{fact:omega-permute}
For any $a, b, c \geq 0$, we have $\omega(a, b, c) = \omega(a, c, b) = \omega(b, a, c) = \omega(b, c, a) = \omega(c, a, b) = \omega(c, b, a)$.
\end{fact}

\begin{fact} \label{fact:omega-scale}
For any $a, b, c, \lambda \geq 0$, we have $\omega(\lambda a, \lambda b, \lambda c) = \lambda \cdot \omega(a, b, c)$.
\end{fact}

\begin{fact} \label{fact:omega-subadditive}
Let $a_i, b_i, c_i \geq 0$. Then $\omega(\sum_i a_i, \sum_i b_i, \sum_i c_i) \leq \sum_i \omega(a_i, b_i, c_i)$.
\end{fact}

\paragraph{Sparse Matrix Multiplication}
This paper is concerned with the \emph{sparse matrix multiplication} problem. For parameters $x, y, z, \IN, \OUT \geq 0$, the input consists of two sparsely-represented matrices $A \in \Int^{x \times y}$ and~\makebox{$B \in \Int^{y \times z}$} such that the total number of nonzero entries in $A$, $B$ is at most~$\IN$ and the total number of nonzero entries in the product $A B$ is at most $\OUT$ (note that we do \emph{not} assume that $\OUT$ is known). The goal is to compute $A B$. We occasionally write~\makebox{$m = \IN + \OUT$}.

Throughout we assume that $x, y, z \leq \IN$ for the following reason: Whenever, say $x > \IN$, then there must be a row $i$ in $A$ which does not contain any nonzero entries. We can safely ignore this row in the multiplication as also the $i$-th row in $A B$ does not contain any nonzero entries. Therefore, after possibly relabeling the rows and columns in $A$ and $B$, we can assume that $x, y, z \leq \IN$.

\paragraph{Machine Model}
While for dense matrix multiplication algorithms it is most natural to consider an algebraic model of computation, in these models it is unclear how to appropriately handle sparse matrices. Therefore, all of our algorithms assume the word RAM model. For integer matrices, we assume that the word size is~$\Order(\log(x y z) + \log(\Delta))$, where $\Delta$ is an upper bound on the largest entry in the matrices $A$ and~$B$ in absolute value. This means that we can perform basic arithmetic and logical operations on indices and entries in constant time. For matrices over general rings $R$, we assume special memory cells storing the ring elements and the ability to perform arithmetic operations over $R$ by an oracle in constant time. No other information on $R$ is needed.

We remark that we defined $\omega(a, b, c)$ in terms of arithmetic circuits, but their complexity translates to the RAM model. Strictly speaking, however, since $\omega(a, b, c)$ is an \emph{infimum},  we only know that for all $\epsilon > 0$, there is a RAM algorithm to multiply $n^a \times n^b$ by $n^b \times n^c$ matrices in time $\Order(n^{\omega(a, b, c) + \epsilon})$. For convenience, many research papers that apply fast matrix multiplication in algorithm design ignore this inaccuracy, but we decided to be explicit in this paper.
% !TEX root = ../paper.tex
\section{Output Densification} \label{sec:densification}
In this section we prove that multiplying input- and output-sparse matrices reduces to multiplying input-sparse rectangular matrices \emph{with dense output}. We refer to this reduction as a \emph{densification}. Specifically, our goal is to prove the following \cref{lem:densification-rings}:

\begin{restatable}[Randomized Densification for Arbitrary Rings]{lemma}{lemdensificationrandom} \label{lem:densification-rings}
Let $R$ be a ring. Suppose there is an algorithm~$\mathcal A$ for matrix multiplication over $R$ with input size $\IN$ in time $T_{\mathcal A}(x, y, z, \IN)$. Then there is a randomized algorithm for matrix multiplication over $R$ with input size $\IN$ and output size $\OUT$ in time:
\begin{equation*}
    2^{\widetilde\Order(\sqrt{\log\IN})} \cdot \max_{\substack{x' \leq x, z' \leq z\\x' \cdot z' \leq \OUT \cdot 2^{\Order(\sqrt{\log \IN})}}} T_{\mathcal A}(x', y, z', \IN).
\end{equation*}
\end{restatable}

We also derive a \emph{deterministic} densification lemma for \emph{nonnegative} matrices (see \cref{lem:densification-nonnegative}). 

Both densification lemmas are proved in two steps: (1) We design the claimed algorithm for sparse matrix multiplication assuming that we know a small superset $S$ of the \emph{support} $\supp(A B) = \set{(i, j) : (A B)[i, j] \neq 0}$. We solve this subtask in \cref{sec:densification:sec:recover} based on hashing and using ideas from sparse recovery. (2) We remove the assumption that $S$ is given by computing $S$ \emph{recursively}. We provide the details of this second step in \cref{sec:densification:sec:support}.

\subsection{Densification via Sparse Recovery} \label{sec:densification:sec:recover}
The goal of this section is to prove the following lemma:

\begin{lemma}[Densification via Sparse Recovery] \label{lem:recover}
Let $R$ be a ring. Suppose there is an algorithm~$\mathcal A$ for matrix multiplication over $R$ with input size $\IN$ in time $T_{\mathcal A}(x, y, z, \IN)$. Then there is an algorithm $\Recover(A, B, S)$ which, given matrices~$A \in R^{x \times y},\, B \in R^{y \times z}$ and a set $S \subseteq [x] \times [z]$ with $\supp(A B) \subseteq S$, computes~$A B$ time
\begin{equation*}
    \widetilde\Order\parens*{\max_{\substack{x' \leq x, z' \leq z\\x' \cdot z' \leq 4|S|}} T_{\mathcal A}(x', y, z', \IN)}.
\end{equation*}
If $\mathcal A$ is deterministic, then $\Recover$ is also deterministic.
\end{lemma}

\begin{algorithm}[t]
\caption{Computes the matrix product $C$ of two sparse matrices $A \in R^{x \times y},\, B \in R^{y \times z}$ over some ring~$R$, given a small superset $S \subseteq [x] \times [z]$ of the support \makebox{$\supp(A B)$}; see \cref{lem:recover}. The only difference between the randomized algorithm and its derandomization is in \cref{alg:recover:line:hash-random,alg:recover:line:hash-det}---in the simple randomized version we use \cref{alg:recover:line:hash-random}a and in the deterministic version we use \cref{alg:recover:line:hash-det}b.} \label{alg:recover}
\begin{algorithmic}[1]
\Procedure{Recover}{$A, B, S$}
    \State Initialize $C \in R^{x \times z}$ to be all-zeros
    \For{$\ell \gets 0, 1, \dots, \floor{\log z}$} \label{alg:recover:line:levels}
        \State Let $A_\ell \in R^{x_\ell \times y}$ denote the restriction of $A$ to indices $i$ with $2^\ell \leq \deg(i) < 2^{\ell+1}$
        \State Let $z_\ell \gets 2^{\ell+2}$
        {
            \makeatletter\def\ALG@step{\refstepcounter{ALG@line}\alglinenumber{\theALG@line a}}\makeatother
            \State Let $\mathcal H_\ell$ be a set of $10 \ceil{\log \IN}$ random hash functions $h : [z] \to [z_\ell]$ \label{alg:recover:line:hash-random}
        }
        {
            \makeatletter\def\ALG@step{\alglinenumber{\theALG@line b}}\makeatother
            \State Let $\mathcal H_\ell$ be the set of hash functions $h : [z] \to [z_\ell]$ computed by \cref{lem:det-hashing} \label{alg:recover:line:hash-det}
        }
        \ForEach{$h \in \mathcal H_\ell$} \label{alg:recover:line:loop}
            \State Let $B' \in R^{y \times z_\ell}$ denote the matrix defined
            \medskip\smallskip
            \Statex[4] $\displaystyle B'[k, b] = \sum_{\substack{j \in [z]\\h(j) = b}} B[k, j]$
            \medskip\smallskip
            \State Compute $C' \gets A_\ell B'$ with the algorithm $\mathcal A$
            \ForEach{\emph{isolated} support element $(i, j) \in S$ (see \cref{def:isolation})}
                \State $C[i, j] \gets C'[i, h(j)]$ \label{alg:recover:line:recover}
            \EndForEach
        \EndForEach
    \EndFor
    \State\Return $C$
\EndProcedure
\end{algorithmic}
\end{algorithm}
    
Consider the pseudocode in \cref{alg:recover} (and ignore \cref{alg:recover:line:hash-det}b for now). Here, for any $i \in [x]$, we define the \emph{degree} $\deg(i) = \abs{\set{j : (i, j) \in S}}$. The algorithm proceeds in \emph{levels}~\makebox{$\ell \gets 0, 1, \dots, \floor{\log z}$} (\cref{alg:recover:line:levels}). At the $\ell$-th level our goal is to recover all entries~$C[i, j]$ where~$i$ has degree roughly~$2^\ell$. More precisely, let $A_\ell$ denote the submatrix of $A$ restricted to indices $i \in [x]$ with~\makebox{$2^\ell \leq \deg(i) < 2^{\ell+1}$} and let $x_\ell$ denote the number of such indices (i.e., the number of rows in $A_\ell$). Our goal at the $\ell$-th level is to recover the submatrix $A_\ell B$ of $C$.

Simply computing $A_\ell B$ is too expensive (as the product of the outer dimensions $x_\ell \cdot z$ can be much larger than $m$). Instead, we will hash the matrix $B$ to a substantially thinner matrix $B' \in \Int^{y \times z_\ell}$ by a random hash function $h : [z] \to [z_\ell]$ where $z_\ell = 2^{\ell+2}$ (in fact, we repeat this step for several random hash functions in \cref{alg:recover:line:hash-random}a and \cref{alg:recover:line:loop}). We define
\begin{equation*}
    B'[k, b] = \sum_{\substack{j \in [z]\\h(j) = b}} B[k, j],
\end{equation*}
and compute the product $C' \gets A_\ell B'$ by means of the fast algorithm $\mathcal A$. To make use of $C'$, we need the following definition:

\begin{definition}[Isolation] \label{def:isolation}
Let $h$ be a hash function. We say that a pair $(i, j) \in S$ is \emph{isolated under $h$} if there is no other $j' \neq j$ with $(i, j') \in S$ and $h(j) = h(j')$. For a set of hash functions $\mathcal H$, we say that~$(i, j) \in S$ is \emph{isolated under $\mathcal H$} if $(i, j)$ is isolated under some $h \in \mathcal H$.
\end{definition}

The algorithm computes the set of isolated pairs $(i, j)$ and for each such pair updates $C[i, j] \gets C'[i, h(j)]$. We prove the correctness of this approach in the next lemma:

\begin{lemma}[Correctness of \cref{alg:recover}]
Assume that for all levels $\ell$, and for all pairs $(i, j) \in S$ with $2^\ell \leq \deg(i) < 2^{\ell+1}$, $(i, j)$ is isolated under $\mathcal H_\ell$. Then \cref{alg:recover} correctly returns $C = A B$.
\end{lemma}
\begin{proof}
We prove that the algorithm correctly returns $C[i, j] = (AB)[i, j]$ for all pairs $(i, j) \in [x] \times [z]$. If~\makebox{$(i, j) \not\in S$}, then we never touch $C[i, j]$ after the initialization to zero. Since $S \supseteq \supp(C)$, the algorithm correctly keeps $C[i, j] = 0$. For the rest of the proof we therefore focus on pairs $(i, j) \in S$, for which we prove the statement in two steps:
\begin{itemize}
    \item\emph{Isolated pairs are correctly recovered:} Focus on some level $\ell$ and some iteration of the inner loop in \cref{alg:recover:line:loop} and suppose that $(i, j)$ satisfies $2^\ell \leq \deg(i) < 2^{\ell+1}$. We show that if $(i, j)$ is isolated under $h$, then we correctly recover $C[i, j] \gets C'[i, h(j)]$ in \cref{alg:recover:line:recover}. The proof is a simple calculation. Recall that by the degree assumption, $i$ is part of the restricted matrix $A_\ell$. Therefore:
    \begin{gather*}
        C'[i, h(j)] = \sum_{k \in [y]} A_\ell[i, k] \cdot B'[k, h(j)] \\
        \qquad= \sum_{k \in [y]} A_\ell[i, k] \cdot \sum_{\substack{j' \in [z]\\h(j) = h(j')}} B[k, j'] \\
        \qquad= \sum_{\substack{j' \in [z]\\h(j) = h(j')}} (A_\ell B)[i, j'].
    \end{gather*}
    Since $(i, j) \in S$ is isolated, any other pair $(i, j') \in \supp(A B) \subseteq S$ must satisfy that $h(j) \neq h(j')$. Therefore, the unique term in the sum is $(A_\ell B)[i, j]$.
    \item\emph{All pairs are correctly recovered.} We use the assumption of the lemma statement: For any $(i, j) \in S$ with $2^\ell \leq \deg(i) < 2^{\ell+1}$ there is some hash function $h \in \mathcal H_\ell$ under which $(i, j)$ is isolated. By the previous bullet, we correctly recover $C[i, j]$ in this iteration. In all iterations where $(i, j)$ is not isolated, we do not change $C[i, j]$. \qedhere
\end{itemize}
\end{proof}
    
To complete the correctness proof, we need to justify the assumption that all pairs $(i, j) \in S$ are isolated. We devote the following two \cref{sec:densification:sec:recover:sec:isolation-random,sec:densification:sec:recover:sec:isolation-det} to this task, and for now focus on the running time.

\begin{lemma}[Running Time of \cref{alg:recover}]
\cref{alg:recover} runs in time
\begin{equation*}
    \widetilde\Order\parens*{\max_{\substack{x' \leq x, z' \leq z\\x' \cdot z' \leq 4|S|}} T_{\mathcal A}(x', y, z', \IN)}.
\end{equation*}
\end{lemma}
\begin{proof}
We can precompute the degrees in time~$\Order(|S|)$ and then compute the matrices~$A_0, \dots, A_{\floor{\log z}}$ in time~$\Order(\IN)$. We again analyze a fixed level $\ell$; the total number of levels is $\log z \leq \log \IN$ and can therefore be neglected. As we will see in \cref{sec:densification:sec:recover:sec:isolation-random,sec:densification:sec:recover:sec:isolation-det}, we can select $\mathcal H$ in time $\widetilde\Order(\IN)$. Moreover, we store all hash functions explicitly and thus evaluating any hash function takes constant time. In each of the $\Order(\log \IN)$ iterations, we construct the matrix~$B'$ in time $\Order(\IN)$ by a single pass over the nonzero entries of $B$. To compute the matrix product $A_\ell \cdot B'$, we use the algorithm~$\mathcal A$. Note that at most a~$2^{-\ell}$-fraction of $[x]$ has degree at least~$2^\ell$, and therefore \smash{$x_\ell \leq \frac{|S|}{2^{\ell}}$}. It follows that~\smash{$x_\ell \cdot z_\ell \leq \frac{|S|}{2^\ell} \cdot 2^{\ell+2} = 4|S|$} and thus the running time of $\mathcal A$ can be bounded by
\begin{equation*}
    T_{\mathcal A}(x_\ell, y, z_\ell, \IN) \leq \max_{\substack{x' \leq x, z' \leq z\\x' \cdot z' \leq 4|S|}} T_{\mathcal A}(x', y, z', \IN).    
\end{equation*}
In time $\Order(|S|)$ we can moreover filter the isolated elements in $S$, and run the recovery loop in \cref{alg:recover:line:loop}. In total, the running time becomes
\begin{equation*}
    \widetilde\Order\parens*{|S| + \max_{\substack{x' \leq x, z' \leq z\\x' \cdot z' \leq 4|S|}} T_{\mathcal A}(x', y, z', \IN)} = \widetilde\Order\parens*{\max_{\substack{x' \leq x, z' \leq z\\x' \cdot z' \leq 4|S|}} T_{\mathcal A}(x', y, z', \IN)};
\end{equation*}
here, the term $|S|$ is dominated since in the worst-case $\mathcal A$ must return an output of size up to~$|S|$.
\end{proof}

\subsubsection{Randomized Isolation} \label{sec:densification:sec:recover:sec:isolation-random}
We quickly show that it is easy to satisfy the required isolation property if $\mathcal H$ is a set of logarithmically many random hash functions (as sampled in \cref{alg:recover:line:hash-random}a).

\begin{lemma}[Randomized Isolation] \label{lem:random-hashing}
Fix any level $\ell$ and let $\mathcal H_\ell$ be a set of $10 \ceil{\log \IN}$ random hash functions $h : [z] \to [z_\ell]$ (as in \cref{alg:recover:line:hash-random}a). With high probability, all pairs $(i, j) \in S$ with $\deg(i) < 2^{\ell+1}$ are isolated under some hash function $h \in \mathcal H$.
\end{lemma}
\begin{proof}
We first focus on a single pair $(i, j) \in S$ and a single random hash function $h : [z] \to [z_\ell]$, and argue that $(i, j)$ is isolated under $h$ with probability at least $\frac12$. Let $J = \set{j' : (i, j') \in S} \setminus \set{j}$ and note that~\makebox{$|J| \leq \deg(i) < 2^{\ell+1}$}. The error event is that there is some $j' \in J$ with $h(j) = h(j')$. For any fixed $j, j'$ this happens with probability at most \smash{$\frac{1}{z_\ell} = \frac{1}{2^{\ell+2}}$}. Taking a union bound over the at most $2^{\ell+1}$ elements in~$J$, we obtain that $(i, j)$ is not isolated with probability at most~\smash{$\frac{2^{\ell+1}}{2^{\ell+2}} = \frac12$}.

Since we pick the $10 \ceil{\log \IN}$ hash functions in $\mathcal H_\ell$ independently, it follows that $(i, j)$ is isolated under~$\mathcal H_\ell$ with probability at least $1 - 2^{-10 \log \IN} = 1 - \IN^{-10}$. By a union bound over the at most $xz \leq \IN^2$ many elements in $S$, we get that all pairs in $S$ are isolated with probability at least $1 - \IN^{-8}$. (The constant $8$ can be chosen arbitrarily larger.)
\end{proof}

\subsubsection{Deterministic Isolation} \label{sec:densification:sec:recover:sec:isolation-det}
In this section we derandomize the previous algorithm. Note that the only place where we use randomness is to guarantee that each pair $(i, j)$ is isolated (at the appropriate level). We will therefore replace the sets~$\mathcal H_\ell$ of random hash functions by hash functions constructed by the following deterministic algorithm: 

\begin{lemma}[Deterministic Isolation] \label{lem:det-hashing}
Let $S \subseteq [x] \times [z]$ be such that for all $i \in [x]$, $\abs{\set{j : (i, j) \in S}} \leq s$ (i.e., each row in $S$ has at most $s$ entries). We can compute a set of hash functions~$\mathcal H \subseteq \set{h : [z] \to [2s]}$ with the following properties:
\begin{enumerate}
    \item All pairs $(i, j) \in S$ are isolated under $\mathcal H$.\\(That is, there is some $h \in \mathcal H$ such that for all~$j' \neq j$ with $(i, j') \in S$, we have $h(j) \neq h(j')$).
    \item $|\mathcal H| \leq \log |S| + 1$.
    \item The algorithm to compute $\mathcal H$ is deterministic and runs in time $\widetilde\Order(z + |S|)$.
\end{enumerate}
\end{lemma}

It is easy to verify that by replacing \cref{alg:recover:line:hash-random}a with \cref{alg:recover:line:hash-det}b in \cref{alg:recover} (and using \cref{lem:det-hashing} in place of \cref{lem:random-hashing} in the analysis), we preserve the correctness of the algorithm. For this reason, in this section we focus on proving \cref{lem:det-hashing}. At the end of this section we further remark how the running time of \cref{alg:recover} is affected.

\paragraph{An Inefficient Algorithm}
\cref{lem:det-hashing} is proven using the method of conditional expectations. We pick the set of hash functions~$\mathcal H$ one by one, and for each hash function $h$ we assign the values $h(1), \dots, h(z)$ one by one. The key idea is that in each step we select the current value $h(i)$ to be \emph{at least as good} as what we expect from a random choice. To precisely state what we mean by \emph{good}, we start with some terminology.

For a set of hash functions $\mathcal H$, let $S_{\mathcal H} \subseteq S$ denote the subset of pairs $(i, j)$ that are \emph{not} isolated under~$\mathcal H$; we will refer to these pairs as \emph{active}. In this language our goal is to find a set of hash functions $\mathcal H$ with zero active pairs. To this end, we start with $\mathcal H \gets \emptyset$ and proceed in several \emph{rounds}. In each round, our goal is to select a hash function that halves the number of active pairs.

Next, we describe how to select such a hash function $h$. Let $C_\mathcal H$ denote the set of all triples $(i, j, j')$ with~\makebox{$j \neq j'$}, \makebox{$(i, j) \in S_{\mathcal H}$} and~\makebox{$(i, j') \in S$}. For a hash function $h : [z] \to [2s] \cup \set\bot$, let $C_{\mathcal H}(h) \subseteq C_{\mathcal H}$ denote the subset of triples~$(i, j, j')$ with $h(j) = h(j') \neq \bot$. We refer to the elements of $C_{\mathcal H}(h)$ as the \emph{collisions under~$h$.} Intuitively, a collision~$(i, j, j')$ is a witness that $(i, j)$ is not isolated under $h$. Our approach is to select a next hash function causing at most $\frac12 |S_{\mathcal H}|$ collisions. More precisely, we start from a hash function initialized to some dummy value $h(1) = \dots = h(z) = \bot$. Then, we enumerate $j \gets 1, \dots, z$ and assign~$h(j) \gets b^*$ for some bucket $b^*$ that causes at most as many new collisions as a random bucket would cause in expectation, i.e.,
\begin{equation*}
    |C_{\mathcal H}(h [j \mapsto b^*])| \leq \Ex_{b \in [2s]} |C_{\mathcal H}(h [j \mapsto b])|.
\end{equation*}
Here, we write $h[j \mapsto b]$ to denote the updated hash function that maps $j$ to $b$. We summarize this algorithm in \cref{alg:det-hashing}, ignoring for now how to find such a bucket $b^*$ (in \cref{alg:det-hashing:line:find-bucket}). Instead, we prove that this algorithm indeed leads to the claimed properties 1 and 2.

\begin{algorithm}[t]
\caption{A deterministic algorithm to find a small set of hash functions $\mathcal H$ isolating the given set $S \subseteq [x] \times [z]$; see \cref{lem:det-hashing}.} \label{alg:det-hashing}
\begin{algorithmic}[1]
\State Let $\mathcal H \gets \emptyset$
\While{$S_{\mathcal H} \neq \emptyset$}
    \State Let $h$ be a function returning $\bot$ for every input
    \For{$j \gets 1, \dots, z$}
        \State Find a bucket $b^*$ with $|C_{\mathcal H}(h [j \mapsto b^*])| \leq \Ex_{b \in [2s]} |C_{\mathcal H}(h [j \mapsto b])|$ \label{alg:det-hashing:line:find-bucket}
        \State $h(j) \gets b^*$
    \EndFor
    \State $\mathcal H \gets \mathcal H \cup \set h$
\EndWhile
\State\Return $\mathcal H$
\end{algorithmic}
\end{algorithm}

\begin{lemma}[Correctness of \cref{alg:det-hashing}] \label{lem:det-hashing-correctness}
\cref{alg:det-hashing} returns a set of hash functions $\mathcal H$ that satisfies properties 1 and 2 of \cref{lem:det-hashing}:
\begin{enumerate}
    \item All pairs $(i, j) \in S$ are isolated under $\mathcal H$.\\(That is, there is some $h \in \mathcal H$ such that for all~$j' \neq j$ with $(i, j') \in S$, we have $h(j) \neq h(j')$).
    \item $|\mathcal H| \leq \log |S| + 1$.
\end{enumerate}
\end{lemma}
\begin{proof}
We prove that the number of active elements $|S_{\mathcal H}|$ halves after each round. From this statement both properties follow immediately: Since initially $|S_\emptyset| = |S|$, this process stops after at most $\log |S| + 1$ rounds which proves property 2. Moreover, as soon as the algorithm terminates we have that $S_{\mathcal H} = \emptyset$ and therefore all pairs in $S$ are isolated under $\mathcal H$, by definition.

Focus on any round; we prove that \cref{alg:det-hashing} selects a hash function $h : [z] \to [2s]$ which causes at most~$\frac12 |S_{\mathcal H}|$ collisions, i.e., $|C_{\mathcal H}(h)| \leq \frac12 |S_{\mathcal H}|$. As any non-isolated active element leads to at least one collision, this entails that indeed the number of active elements halves after this round, i.e.,~\smash{$|S_{\mathcal H \cup \set h}| \leq \frac12 |S_{\mathcal H}|$}.

Consider any triple $(i, j, j') \in C_{\mathcal H}$; when do we decide whether $(i, j, j')$ becomes a collision? We claim that this decision happens in exactly the $\ell = \max\set{j, j'}$-th iteration of the inner loop. Before that iteration we have not yet fixed the hash value $h(\max\set{j, j'})$ and after that iteration we will never change the hash values $h(j)$ and $h(j')$. Let \smash{$C_{\mathcal H}^{(\ell)}$} denote the subset of triples $(i, j, j')$ with $\max\set{j, j'} = \ell$, and let \smash{$C_{\mathcal H}^{(\ell)}(h) = C_{\mathcal H}^{(\ell)} \cap C_{\mathcal H}(h)$}. With this insight in mind, focus on the $j$-th iteration of the inner loop. For any bucket $b$, we have that
\begin{equation*}
    |C_{\mathcal H}(h[j \mapsto b])| = |C_{\mathcal H}(h)| + |C_{\mathcal H}^{(j)}(h[j \mapsto b])|.
\end{equation*}
Therefore, we select a bucket $b^*$ with
\begin{gather*}
    |C_{\mathcal H}(h)| + |C_{\mathcal H}^{(j)}(h[j \mapsto b^*])| \leq \Ex_{b \in [2s]} \parens*{|C_{\mathcal H}(h)| + |C_{\mathcal H}^{(j)}(h[j \mapsto b])|} \\
    \qquad= |C_{\mathcal H}(h)| + \Ex_{b \in [2s]} \parens*{|C_{\mathcal H}^{(j)}(h[j \mapsto b])|} \\
    \qquad= |C_{\mathcal H}(h)| + \frac{|C_{\mathcal H}^{(j)}|}{2s}.
\end{gather*}
In the last step we have used that for any triple $(i, j, j') \in C_{\mathcal H}^{(\ell)}$, there is exactly one choice for $h(\ell)$ which turns this triple into a collision (namely, $h(\ell) \gets h(\min\set{j, j'})$). This calculation shows that the number of \emph{new} collisions in the $j$-th round is at most \smash{$\frac{1}{2s} |C_{\mathcal H}^{(j)}|$}. Thus, the total number of collisions for the hash function $h$ selected in the current round is
\begin{equation*}
    |C_{\mathcal H}(h)| \leq \sum_{j=1}^z \frac{|C_{\mathcal H}^{(j)}|}{2s} = \frac{|C_{\mathcal H}|}{2s} \leq \frac{|S_{\mathcal H}| \cdot s}{2s} = \frac{|S_{\mathcal H}|}{2}.
\end{equation*}
This is precisely what we set out to prove.
\end{proof}

\paragraph{How To Efficiently Select Good Buckets}
In the previous lemma we have completely analyzed \cref{alg:det-hashing}---except that we left open \emph{how} to find the good buckets $b^*$ in \cref{alg:det-hashing:line:find-bucket}. It is possible to simply test all options~\makebox{$b \gets 1, \dots, 2s$}, but this is too inefficient for our purposes. We now describe a faster approach leading to the claimed running time. The idea is to find $b^*$ via binary search. We maintain the following data:
\begin{itemize}
    \item A matrix $M \in \Nat^{x \times 2s}$ where $M[i, b] = \abs{\set{(i, j) \in S : h(j) = b}}$.\\(The entry $M[i, b]$ counts for how many pairs $(i, j)$ we have hashed $j$ to the bucket $b$.)
    \item A matrix $M_{\mathcal H} \in \Nat^{x \times 2s}$ where $M_{\mathcal H}[i, b] = \abs{\set{(i, j) \in S_{\mathcal H} : h(j) = b}}$.\\(The entry $M_{\mathcal H}[i, b]$ counts for how many \emph{active} pairs $(i, j)$ we have hashed $j$ to the bucket $b$.)
    \item Moreover, we maintain segment trees on the rows of the matrices $M$ and $M_{\mathcal H}$ such that, given $i \in [x]$ and an interval~$B \subseteq [2s]$, we can compute $\sum_{b \in B} M[i, b]$ and $\sum_{b \in B} M_{\mathcal H}[i, b]$ in polylogarithmic time.
\end{itemize}
This data is easy to maintain: As the hash function is initialized to $\bot$, we initially set both matrices to be all-zeros. Whenever we assign $h(j) \gets b^*$ we increment all entries $M[i, b]$ where $(i, j) \in S$ and all entries~$M_{\mathcal H}[i, b]$ where $(i, j) \in S_{\mathcal H}$. We update the segment tree accordingly. The implementation of the binary search hinges on the following lemma:

\begin{lemma} \label{lem:det-hashing-identity}
Focus on any round of \cref{alg:det-hashing}. In the $j$-th iteration of the inner loop, we have that
\begin{equation*}
    |C_{\mathcal H}(h[j \mapsto b])| = |C_{\mathcal H}(h)| + \sum_{\substack{i \in [x]\\(i, j) \in S_{\mathcal H}}} M[i, b] + \sum_{\substack{i \in [x]\\(i, j) \in S}} M_{\mathcal H}[i, b].
\end{equation*}
\end{lemma}
\begin{proof}
Focus on any round of \cref{alg:det-hashing}. In the $\ell$-th iteration of the inner loop we have already assigned~\makebox{$h(1), \dots, h(\ell-1) \in [2s]$}, but the values $h(\ell), \dots, h(z)$ are still set to $\bot$. The proof is by the following simple yet tedious calculation:
\begin{gather*}
    |C_{\mathcal H}(h[\ell \mapsto b])| = |C_{\mathcal H}(h)| + |C^{(\ell)}_{\mathcal H}(h[\ell \mapsto b])| \\
    \qquad= |C_{\mathcal H}(h)| + \abs{\set{(i, j, j') \in C_{\mathcal H}^{(\ell)} : h(\min\set{j, j'}) = b}} \\
    \qquad=
    \begin{multlined}[t]
        |C_{\mathcal H}(h)| + \abs{\set{(i, j, j') \in C_{\mathcal H} : j' < j = \ell, h(j') = b}} \\+ \abs{\set{(i, j, j') \in C_{\mathcal H} : j < j' = \ell, h(j) = b}}
    \end{multlined} \\
    \qquad= |C_{\mathcal H}(h)| + \sum_{\substack{i \in [x]\\(i, \ell) \in S_{\mathcal H}}} \abs{\set{(i, j') \in S : h(j') = b}} + \sum_{\substack{i \in [x]\\(i, \ell) \in S}} \abs{\set{(i, j) \in S_{\mathcal H} : h(j) = b}} \\
    \qquad= |C_{\mathcal H}(h)| + \sum_{\substack{i \in [x]\\(i, \ell) \in S_{\mathcal H}}} M[i, b] + \sum_{\substack{i \in [x]\\(i, \ell) \in S}} M_{\mathcal H}[i, b]. \qedhere
\end{gather*}
\end{proof}

\begin{lemma}[Running Time of \cref{alg:det-hashing}]
There is an implementation of \cref{alg:det-hashing:line:find-bucket} in \cref{alg:det-hashing} such that the total running time becomes $\widetilde\Order(z + |S|)$.
\end{lemma}
\begin{proof}
Focus on any round of \cref{alg:det-hashing}, and the $j$-th iteration of the inner loop. We find a bucket $b^*$ satisfying $|C_{\mathcal H}(h [j \mapsto b^*])| \leq \Ex_{b \in [2s]} |C_{\mathcal H}(h [j \mapsto b])|$ by binary search as follows. We maintain a search interval $B$ and the invariant
\begin{equation*}
    \Ex_{b \in B} |C_{\mathcal H}(h[j \mapsto b])| \leq \Ex_{b \in [2s]} |C_{\mathcal H}(h[j \mapsto b])|.
\end{equation*}
Initially, $B = [2s]$ and thus the invariant holds trivially. In each step, we split $B = B_1 \cup B_2$ into roughly equal-sized parts and recur on the half $k$ with the smaller value $\Ex_{b \in B_k} |C_{\mathcal H}(h[j \mapsto b])|$. It is clear that this half maintains the invariant. To compute these values, we use that for any interval we can express
\begin{gather*}
    \Ex_{b \in B} |C_{\mathcal H}(h[j \mapsto b])| = \frac1{|B|}\sum_{b \in B} |C_{\mathcal H}(h[j \mapsto b])| \\
    \qquad= |C_{\mathcal H}(h)| + \frac{1}{|B|} \sum_{\substack{i \in [x]\\(i, j) \in S_{\mathcal H}}} \sum_{b \in B} M[i, b] + \frac{1}{|B|} \sum_{\substack{i \in [x]\\(i, j) \in S}} \sum_{b \in B} M_{\mathcal H}[i, b]
\end{gather*}
by \cref{lem:det-hashing-identity}. The first term in the sum is the same for both halves $B_1$ and $B_2$ and can be ignored. The second and third terms can be computed in time $\abs{\set{i : (i, j) \in S}} \cdot \polylog(s)$, using that the inner sums can be evaluated in polylogarithmic time by the segment trees.

Let us analyze the total running time. The algorithm runs for $\Order(\log |S|)$ rounds. In each round, we run~$z$ iterations, where the $j$-th iteration takes time $\Order(1 + \abs{\set{i : (i, j) \in S}} \cdot \polylog(s))$. Therefore, across all iterations the running time is bounded by \smash{$\Order(z + |S| \polylog(s)) = \widetilde\Order(z + |S|)$}. The total number of updates to the matrices $M$ and $M_{\mathcal H}$ is $|S|$ each, and therefore the overhead due to maintaining $M$, $M_{\mathcal H}$ and the segment tree becomes \smash{$\widetilde\Order(|S|)$}, too.
\end{proof}

We remark that the application of \cref{lem:det-hashing} in the recovery algorithm (\cref{alg:recover}) incurs a running time overhead of \smash{$\sum_\ell \widetilde\Order(z + x_\ell z_\ell) = \sum_\ell \widetilde\Order(z + \frac{|S|}{2^\ell} \cdot 2^{\ell}) = \widetilde\Order(\IN + |S|)$}. This does not increase the running time asymptotically.

\subsection{Approximating the Support} \label{sec:densification:sec:support}
We are ready to complete the proofs of the densification lemmas. By the previous section it suffices to compute a small over-approximation of the support $\supp(A B)$, and the main goal of this section is to prove that this approximation can be computed recursively.

\paragraph{Warm-Up: Nonnegative Integer Matrices}
We start with the following easy deterministic densification for nonnegative integer matrices (in particular, for Boolean matrices). The advantage of nonnegative matrices is that we can conveniently avoid cancellations. 

\begin{lemma}[Deterministic Densification for Nonnegative Matrices] \label{lem:densification-nonnegative}
Suppose there is a deterministic algorithm $\mathcal A$ for nonnegative integer matrix multiplication with input size $\IN$ in time $T_{\mathcal A}(x, y, z, \IN)$. Then there is a deterministic algorithm for nonnegative integer matrix multiplication with input size $\IN$ and output size $\OUT$ in time:
\begin{equation*}
    \widetilde\Order\parens*{\max_{\substack{x' \leq x, z' \leq z\\x' \cdot z' \leq 8\OUT}} T_{\mathcal A}(x', y, z', \IN)}.
\end{equation*}
\end{lemma}

\begin{algorithm}[t]
\caption{A deterministic algorithm to multiply sparse nonnegative matrices $A \in \Nat^{x \times y},\, B \in \Nat^{y \times z}$; see \cref{lem:densification-nonnegative}.} \label{alg:densification-det}
\begin{algorithmic}[1]
\Procedure{MultiplyNonnegative}{$A, B$}
    \If{$x \leq 1$}
        \State Solve the problem naively in time $\Order(\IN)$
    \Else
        \State Let $A'$ be the $\ceil{\frac x2} \times y$ matrix defined by $A'[i, k] = A[2i-1, k] + A[2i, k]$
        \State Compute $C' \gets \MultiplyNonnegative(A', B)$ recursively
        \State Compute $S \gets \set{(2i-1, j), (2i, j) : (i, j) \in \supp(C')}$
        \State\Return \Call{Recover}{$A, B, S$}
    \EndIf
\EndProcedure
\end{algorithmic}
\end{algorithm}

\begin{proof}
Throughout, let $A \in \Nat^{x \times y}$ and $B \in \Nat^{y \times z}$ denote the given matrices. Our goal is to compute their product $C = A B$ by an algorithm \Call{MultiplyNonnegative}{$A, B$}; see the pseudocode in \cref{alg:densification-det}. If~\makebox{$x \leq 1$}, then the problem deforms to a vector-matrix multiplication that is trivially solvable in linear time~$\Order(\IN)$. So assume that $x > 1$. By \cref{lem:recover} it suffices to compute a set $S \subseteq [x] \times [z]$ satisfying the following two properties:
\begin{enumerate}[label=(\roman*)]
    \item $S \supseteq \supp(C)$, and
    \item $|S| \leq 2|\supp(C)| \leq 2\OUT$.
\end{enumerate}
We will compute this set~$S$ by calling $\MultiplyNonnegative$ recursively. Specifically, let $A'$ be the~$\ceil{\frac x2} \times y$ matrix defined by $A'[i, k] = A[2i - 1, k] + A[2i, k]$. (In case that there is an overflow and $2i > x$, simply treat~$A[2i, k]$ as zero.) We recursively compute~$C' = A' B$, and let $S = \set{(2i - 1, j), (2i, j) : (i, j) \in \supp(C')}$. The interesting part is to prove the two properties.
\begin{enumerate}[label=(\roman*)]
    \item By definition we have that
    \begin{gather*}
        C'[i, j] = \sum_{k \in [y]} A'[i, k] \cdot B[k, j] \\
        \qquad= \sum_{k \in [y]} (A[2i - 1, k] + A[2i, k]) \cdot B[k, j] \\
        \qquad= C[2i - 1, j] + C[2i, j].
    \end{gather*}
    Hence, whenever $(2i - 1, j) \in \supp(C)$ or $(2i, j) \in \supp(C)$ then have $C'[i, j] > 0$---here, we use that the matrices are nonnegative. In this case, it follows that $(i, j) \in \supp(C')$ and that \makebox{$(2i - 1, j), (2i, j) \in S$}.
    \item From the same calculation one can see that $(i, j) \in \supp(C')$ only if at least one of the pairs $(2i - 1, j)$ and $(2i, j)$ is part of $\supp(C)$. It follows that $|S| = 2|\supp(C')| \leq 2|\supp(C)| \leq 2\OUT$, which proves property (ii).
\end{enumerate}
It remains to analyze the running time. The running time is dominated by the call to $\Recover$ with running time bounded by
\begin{equation*}
    \widetilde\Order\parens*{|S| + \max_{\substack{x' \leq x, z' \leq z\\x' \cdot z' \leq 4|S|}} T_{\mathcal A}(x', y, z', \IN)} = \widetilde\Order\parens*{\max_{\substack{x' \leq x, z' \leq z\\x' \cdot z' \leq 8\OUT}} T_{\mathcal A}(x', y, z', \IN)},
\end{equation*}
where the term $|S| = \Order(\OUT)$ is dominated by the running time of $\mathcal A$. $\MultiplyNonnegative$ reaches a recursion depth of at most~$\Order(\log x) = \Order(\log \IN)$ which only incurs a logarithmic factor.
\end{proof}

\paragraph{Integer Matrices}
This approach is promising to extend to integers and arbitrary rings, but there is a difficulty: Suppose that in the product matrix $C = AB$ there are two entries $C[2i - 1, j] = -C[2i, j]$. Then, if we compute $A'$ and $C'$ as before, these two entries \emph{cancel} leading to $C'[i, j] = 0$. We have therefore lost the support elements $(2i - 1, j), (2i, j)$ and the set $S$ is incorrect. While it is hard to deal with these cancellations deterministically, we show how to remedy this state of affairs by exploiting randomization.

\begin{lemma}[Randomized Densification for Integer Matrices] \label{lem:densification-integer}
Suppose there is an algorithm $\mathcal A$ for integer matrix multiplication with input size $\IN$ in time $T_{\mathcal A}(x, y, z, \IN)$. Then there is a randomized algorithm for integer matrix multiplication with input size $\IN$ and output size $\OUT$ in time:
\begin{equation*}
    \widetilde\Order\parens*{\max_{\substack{x' \leq x, z' \leq z\\x' \cdot z' \leq 8\OUT}} T_{\mathcal A}(x', y, z', \IN)}.
\end{equation*}
\end{lemma}
\begin{proof}
This proof is almost exactly as in \cref{lem:densification-nonnegative}, except for the following difference: We define the matrix $A'$ by $A'[i, k] = A[2i - 1, k] + r \cdot A[2i, k]$, where $r$ is a random integer in the range $[\IN^{10}]$. The only error event is that there are nonzero entries $C[2i - 1, j]$ and $C[2i, j]$ such that $C[2i - 1, j] + r \cdot C[2i, j] = 0$ (in which case we incorrectly miss the support elements $(2i - 1, j)$ and $(2i, j)$ in $S$). For any such pair, there is at most one bad value of $r$ and thus, this event happens with probability at most $\IN^{-10}$. Taking a union bound over the at most $xz \leq \IN^2$ pairs, the error event happens with small probability $\IN^{-8}$. (The constant~$8$ can be arbitrarily larger.)

There is a subtle consequence: With each recursive call, we increase the entries in $A$ by a factor $\poly(\IN)$. Therefore, at the deepest level of the recursion, the numbers have increased by a factor \smash{$\IN^{\Order(\log \IN)}$}. To represent such numbers, we need $\Order(\log \IN)$ memory cells per number, which worsens the time complexity by a polylogarithmic factor.
\end{proof}

\paragraph{Matrices over Arbitrary Rings}
So far, we have exploited the specific structure of the integers (namely, that any univariate linear equation has at most one integer solution, which is most likely not the randomly chosen value $r$). For arbitrary rings, we can neither assume any structure, nor that we have access to random ring elements. Densification is still possible, using the following simple key lemma:

\begin{lemma} \label{lem:ring-random-zero}
Let $R$ be a ring and let $a_1, \dots, a_w \in R$ at least one of which is nonzero. For a uniformly random subset $I \subseteq [w]$, the probability that $\sum_{i \in I} a_i = 0$ is at most $\frac12$.    
\end{lemma}
\begin{proof}
Without loss of generality, assume that $a_1 \neq 0$. Writing $I_1 = I \cap \set{1},\, I_2 = I \cap \set{2, \dots, w}$, note that the random sets $I_1$ and $I_2$ are independently selected. We therefore consider $I_2$ fixed and only rely on the randomness of $I_1$. The condition $\sum_{i \in I} a_i$ can be equivalently rewritten as $a_1 |I_1| = -\sum_{i \in I_2} a_i$ where $|I_1|$ is either $0$ or $1$ with equal probabilities \smash{$\frac12$}.  We thus distinguish two cases: If $-\sum_{i \in I_2} a_i = 0$, then we succeed in the event that $I_1$ is nonempty, and if $-\sum_{i \in I_2} a_i \neq 0$, then we succeed in the event that $I_1$ is empty.
\end{proof}

This lemma leads to the following densification result for arbitrary rings:

\lemdensificationrandom*

\begin{algorithm}[t]
\caption{A randomized algorithm to multiply sparse matrices $A \in R^{x \times y},\, B \in R^{y \times z}$ for general rings~$R$; see \cref{lem:densification-rings}. In the algorithm, we use the parameter \smash{$w = 2^{\sqrt{\log x}}$}.} \label{alg:densification-random}
\begin{algorithmic}[1]
\Procedure{Multiply}{$A, B$}
    \If{$x \leq 1$}
        \State Solve the problem naively in time $\Order(\IN)$
    \Else
        \For{$\ell \gets 1, \dots, L := 10 \ceil{\log(\IN)}$}
            \State Sample a random subset $I_\ell \subseteq [w]$
            \State Let $A_\ell$ be the $\ceil{\frac{x}{w}} \times y$ matrix defined by
            \medskip
            \Statex[4] $\displaystyle A_\ell[i, k] = \sum_{i' \in I_\ell} A[(i - 1)w + i', k]$
            \medskip
            \State $C_\ell \gets \textsc{Multiply}(A_\ell, B)$
            \State $S_\ell \gets \set{((i-1)w + i', j) : (i, j) \in \supp(C_\ell),\, i' \in [w]}$
        \EndFor
        \State $S \gets S_1 \cup \dots \cup S_L$
        \State\Return \Call{Recover}{$A, B, S$}
    \EndIf
\EndProcedure
\end{algorithmic}
\end{algorithm}

\begin{proof}
The idea is very similar to \cref{lem:densification-nonnegative}. We design an algorithm $\Multiply$; see the pseudocode in \cref{alg:densification-random}. If $x \leq 1$, then the problem is trivially solvable in time $\Order(\IN)$, so suppose otherwise. Let $w$ be a parameter to be fixed later. As before, our goal is to compute a set $S \subseteq [x] \times [z]$ satisfying the following two properties so that the matrix product can be computed by $\Recover$ (\cref{lem:recover}):
\begin{enumerate}[label=(\roman*)]
    \item $S \supseteq \supp(C)$ with high probability, and
    \item $|S| \leq w \cdot |\supp(C)| \leq w \OUT$.
\end{enumerate}
To compute $S$, we repeat the following steps for $\ell \gets 1, \dots, L := 10 \ceil{\log(\IN)}$: Sample a uniformly random subset $I_\ell \subseteq [w]$, and let $A_\ell$ be the $\ceil{\frac xw} \times y$ matrix defined by $A'[i, k] = \sum_{i' \in I_\ell} A[(i - 1)w + i', k]$ (again we ignore overflows and assume that $A$ is zero outside of its bounds). By calling $\Multiply$ recursively, we compute the matrix product $C' = A' B$. Then we pick $S_\ell = \set{((i-1) w + i', j) : (i, j) \in \supp(C'), i' \in [w]}$. Finally, take $S = S_1 \cup \dots \cup S_L$. We show that both properties (i) and (ii) are satisfied by this construction:
\begin{enumerate}[label=(\roman*)]
    \item Consider any support element from $\supp(C)$, written as $((i - 1) w + i', j)$ for $i \in [x],\, i' \in [w],\, j \in [z]$. We prove that with high probability, $((i - 1)w + i', j) \in S$. First, focus on any repetition $\ell$. By the definition of $A_\ell$ and $C_\ell$, we have
    \begin{gather*}
        C_\ell[i, j] = \sum_{k \in [y]} A_\ell[i, k] \cdot B[k, j] \\
        \qquad= \sum_{k \in [y]} \parens*{\sum_{i' \in I_\ell} A[(i-1)w + i']} \cdot B[k, j] \\
        \qquad= \sum_{i' \in I_\ell} C[(i-1)w + i', j].
    \end{gather*}
    Recall that $I_\ell \subseteq [w]$ is a uniformly random subset and that at least one of the terms in the sum is nonzero (namely, $C[(i - 1)w + i', j]$). Therefore, by \cref{lem:ring-random-zero}, $C_\ell[i, j]$ is nonzero with probability at most $\frac12$. Recall that in this case we include $((i - 1) + i', j)$ in $S_\ell$. It follows that $((i - 1) + i', j) \in S$ with high probability \smash{$1 - 2^{-L} \geq 1 - \IN^{-10}$}.
    \item From the same calculation we conclude that $(i, j) \in \supp(C_\ell)$ only if there is some $i' \in [w]$ such that $((i - 1)w + i', j) \in \supp(C)$. It follows that all sets $S_1, \dots, S_L$ (and thus also $S$) are contained in a superset of $\supp(C)$ of size at most $w \cdot |\supp(C)| = w \OUT$. 
\end{enumerate}
This completes the correctness proof, but it remains to analyze the running time. Ignoring the cost of the recursive calls, the running time of $\Multiply$ is dominated by calling $\Recover$. Each such call runs in time
\begin{equation*}
    \widetilde\Order\parens*{|S| + \max_{\substack{x' \leq x, z' \leq z\\x' \cdot z' \leq 4|S|}} T_{\mathcal A}(x', y, z', \IN)} = \widetilde\Order\parens*{\max_{\substack{x' \leq x, z' \leq z\\x' \cdot z' \leq 4w \OUT}} T_{\mathcal A}(x', y, z', \IN)}.
\end{equation*}
However, this time the recursive calls are considerably more costly than in \cref{lem:densification-nonnegative}: The recursion reaches depth $\Order(\log_w(x)) = \Order(\log_w(\IN))$ and each node in the recursion tree branches with degree $L = \Order(\log \IN)$. It follows that the number of nodes in the recursion tree is bounded by $(\log \IN)^{\Order(\log_w \IN)}$. Thus, by setting~\smash{$w = 2^{\sqrt{\log \IN}}$}, the total running time is
\smallskip
\begin{gather*}
    (\log \IN)^{\Order(\log_w \IN)} \cdot \!\!\!\max_{\substack{x' \leq x, z' \leq z\\x' \cdot z' \leq 4w \OUT}} T_{\mathcal A}(x', y, z', \IN) \\
    \qquad= 2^{\Order(\sqrt{\log \IN} \log\log \IN)} \cdot \!\!\!\!\!\!\!\!\!\max_{\substack{x' \leq x, z' \leq z\\x' \cdot z' \leq \OUT \cdot 2^{\sqrt{\log \IN} + 2}}} T_{\mathcal A}(x', y, z', \IN),
\end{gather*}
which is as claimed.
\end{proof}
% !TEX root = ../paper.tex
\section{The Complete Algorithm} \label{sec:input-sparse}
In this section we complete our algorithm for sparse matrix multiplication. We structure this section as follows: We first give the heavy/light algorithm for input-sparse matrix multiplication in \cref{sec:input-sparse:sec:algo}, and combine this algorithm with the densification lemmas to obtain our main results in \cref{sec:input-sparse:sec:main-theorems}. Both algorithms rely on our sparse matrix multiplication exponent $\sigma(r)$ as defined in the following definition:\footnote{Strictly speaking, $\sigma(r)$ also depends on the underlying ring. That is, for a ring $R$ we define $\sigma_R(r)$ as the unique solution to the equation $\omega_R(\sigma - 1, 2 - \sigma, 1 + r - \sigma) = \sigma$. As for $\omega$, we typically omit the subscript when the ring $R$ is clear.}

\defsparseexp*

In \cref{sec:input-sparse:sec:sparse-exp} we analyze this exponent $\sigma(r)$ in detail: We argue that $\sigma(r)$ is well-defined, and derive readable upper bounds on $\sigma(r)$.

\subsection{Input-Sparse Matrix Multiplication} \label{sec:input-sparse:sec:algo}
The goal of this section is to design a heavy/light algorithm for input-sparse matrix multiplication. The analysis of this algorithm (in the heavy case) hinges on the following lemma:

\begin{lemma} \label{lem:rectangular-hardest}
Let $a, b, c \geq 0$ with $a \leq c$. Then, for any $0 \leq \delta \leq \frac{c-a}{2}$, it holds that $\omega(a + \delta, b, c - \delta) \leq \omega(a, b, c)$.
\end{lemma}
\begin{proof}
If $a = c$, then also $\delta = 0$ and the claim is trivial. So assume that $c - a > 0$ and let $\epsilon = \frac{\delta}{c-a}$. By \cref{fact:omega-permute,fact:omega-scale}, we have that $\omega((1-\epsilon) a, (1-\epsilon) b, (1 - \epsilon) c) = (1 - \epsilon) \omega(a, b, c)$ and $\omega(\epsilon c, \epsilon b, \epsilon a) = \epsilon \omega(c, b, a) = \epsilon \omega(a, b, c)$. Thus, by \cref{fact:omega-subadditive} it follows that
\begin{gather*}
    \omega(a + \delta, b, c - \delta) = \omega(a + \epsilon (c - a), b, c - \epsilon (c - a)) \\
    \qquad= \omega((1 - \epsilon) a + \epsilon c, (1 - \epsilon) b + \epsilon b, (1 - \epsilon) c + \epsilon a) \\
    \qquad\leq \omega((1 - \epsilon) a, (1 - \epsilon) b, (1 - \epsilon) c) + \omega(\epsilon c, \epsilon b, \epsilon a) \\
    \qquad= (1 - \epsilon) \omega(a, b, c) + \epsilon \omega(a, b, c) \\
    \qquad= \omega(a, b, c),
\end{gather*}
which proves the claim.
\end{proof}

\begin{lemma}[Deterministic Input-Sparse Matrix Multiplication] \label{lem:input-sparse}
Let $R$ be a ring and let $r \in [0, 2]$. Given two matrices $A \in R^{x \times y}$ and $B \in R^{y \times z}$ with input size $\IN$ and \smash{$x z \leq \IN^r$}, we can compute their product $AB$ is in deterministic time \smash{$\Order(\IN^{\sigma(r)+\epsilon})$}, for any $\epsilon > 0$.
\end{lemma}
\begin{proof}
We assume by symmetry that $x \leq z$ (otherwise we use the identity~$A B = (B^T A^T)^T$ to compute $A B$ and thereby exchange $x$ and $z$). Throughout, let $\sigma = \sigma(r)$ (and recall that $\sigma \geq 1$).

We apply a heavy-light approach: We say that $k \in [y]$ is \emph{light} if \makebox{$\abs{\set{i : A[i, k] \neq 0}} \leq \IN^{\sigma-1}$}, and \emph{heavy} otherwise. Let $y_1$ and~$y_2$ denote the number of light and heavy indices, respectively. We subdivide~$A$ into submatrices~\makebox{$A_1 \in \Int^{x \times y_1}$} and~\makebox{$A_2 \in \Int^{x \times y_2}$}, where the light indices participate in $A_1$ and the heavy indices participate in $A_2$. Subdivide~$B$ similarly into $B_1 \in \Int^{y_1 \times z}$ and~\makebox{$B_2 \in \Int^{y_2 \times z}$}. Then the algorithm executes the following two steps:
\begin{enumerate}
\item Compute $r_1 = A_1 \cdot B_1$ exploiting the sparsities of $A_1$ and $B_1$. Specifically:
\smallskip
\begin{algorithmic}[1]
\makeatletter\def\ALG@step{\normalsize--}\makeatother
\State Initialize $r_1$ as an all-zero $x \times z$ matrix
\For{$(k, j) \in [y_1] \times [z]$ with $B_1[k, j] \neq 0$}
    \For{$i \in [x]$ with $A_1[i, k] \neq 0$}
        \State $r_1[i, j] \gets r_1[i, j] + A_1[i, k] \cdot B_1[k, j]$
    \EndFor
\EndFor
\end{algorithmic}
\item Compute $r_2 = A_2 \cdot B_2$ using fast rectangular matrix multiplication (ignoring the assumption that $A_2$ and $B_2$ are sparse).
\end{enumerate}
We report $C = r_1 + r_2$ as the output.

\medskip
The correctness is fairly obvious: For any subdivision of $[y]$ we have that $A \cdot B = A_1 \cdot B_1 + A_2 \cdot B_2$. It is easy to verify that step 1 correctly computes $r_1 = A_1 \cdot B_1$ and immediate that step 2 correctly computes $r_2 = A_2 \cdot B_2$.

\medskip
The interesting part is to analyze the running time. We start with step~1. The outer loop runs for at most $\IN$ iterations (since $B$ contains at most $\IN$ nonzero entries). And for any~$k$, the inner loop runs for at most~\smash{$\IN^{\sigma-1}$} iterations since $A_1$ corresponds to the light indices. Therefore, the running time of step~1 is bounded by~\smash{$\Order(\IN^{\sigma})$}.

For the second step, we start with some observations: First, if there is no heavy index (i.e., $y_2 = 0$) then step~2 is trivial. We may therefore assume that there is at least one heavy index which implies that~\smash{$x \geq \IN^{\sigma-1}$}. Second, using the assumptions that $x z \leq \IN^r$ and $x \leq z$, there is some \smash{$0 \leq \delta \leq 1 + \frac{r}{2} - \sigma$} such that \smash{$x \leq \IN^{\sigma - 1 + \delta}$} and~\smash{$z \leq \IN^{1 + r - \sigma - \delta}$}. Third, the number of heavy indices is bounded by \smash{$y_2 \leq \IN / \IN^{\sigma-1} = \IN^{2-\sigma}$}. Recall that we compute the matrix product~$r_2 = A_2 \cdot B_2$ using fast matrix multiplication. The running time of this step is thus~\smash{$\Order(\IN^{\omega(\sigma-1+\delta, 2-\sigma, 1+r-\sigma-\delta)+\epsilon})$}, for all $\epsilon > 0$. This exponent can be bounded using \cref{lem:rectangular-hardest}:
\begin{equation*}
    \omega(\sigma - 1 + \delta, 2 - \sigma, 1 + r - \sigma - \delta) \leq \omega(\sigma - 1, 2 - \sigma, 1 + r - \sigma) = \sigma.
\end{equation*}
In the final step we have applied the definition of $\sigma = \sigma(r)$. In summary, the second step runs in time \smash{$\Order(\IN^{\sigma+\epsilon})$}, which completes the proof.
\end{proof}

\subsection{Sparse Matrix Multiplication} \label{sec:input-sparse:sec:main-theorems}
We obtain our main results by a combination of the densification technique from the previous \cref{sec:densification} with the input-sparse algorithm. Specifically, the proofs of our Main \cref{thm:sparse-mm-boolean,thm:integer-bivariate} are a straightforward combination of \cref{lem:input-sparse} with the densification algorithms from \cref{lem:densification-nonnegative,lem:densification-rings}:

\thmsparsemmboolean*

\thmsparsemminteger*

We will later argue that these theorems are conditionally optimal, but at this point it is far from clear how to even interpret the running time \smash{$\IN^{\sigma(r)+\epsilon}$}. Therefore, the next section is devoted to a detailed analysis of~$\sigma(r)$.

\subsection{The Sparse Matrix Multiplication Exponent} \label{sec:input-sparse:sec:sparse-exp}
\paragraph{The Sparse Matrix Multiplication Exponent Is Well-Defined}
Recall that we define $\sigma(r)$ as the unique solution to the equation $\omega(\sigma - 1, 2 - \sigma, 1 + r - \sigma) = \sigma$. However, it is a priori not clear that a solution~$\sigma$ to this equation even exists. Therefore, our first goal is to prove that $\sigma(r)$ is well-defined in the sense that there indeed exists unique solution (see \cref{lem:sparse-exp-well-defined}).

\begin{lemma} \label{lem:sparse-exp-increasing}
For any $r \in [0, 2]$, the function $\sigma - \omega(\sigma - 1, 2 - \sigma, 1 + r - \sigma)$ is strictly increasing for $\sigma \in [1, 1+\frac r2]$.
\end{lemma}
\begin{proof}
Let $f(\sigma) = \omega(\sigma - 1, 2 - \sigma, 1 + r - \sigma)$. To prove the claim it suffices to show that $f$ is nonincreasing. And indeed, for any $\sigma, \sigma' \in [1, 1 + \frac{r}{2}]$ with $\sigma \leq \sigma'$ we have
\begin{gather*}
    \omega(\sigma' - 1, 2 - \sigma', 1 + r - \sigma') \\
    \qquad\leq \omega(\sigma' - 1, 2 - \sigma, 1 + r - \sigma') \\
    \qquad= \omega(\sigma - 1 + (\sigma' - \sigma), 2 - \sigma, 1 + r - \sigma - (\sigma' - \sigma)) \\
    \qquad\leq \omega(\sigma - 1, 2 - \sigma, 1 + r - \sigma).
\end{gather*}
Here, in the last step, we have applied \cref{lem:rectangular-hardest} with $\delta = \sigma' - \sigma$; note that this choice satisfies the required condition \smash{$\delta \leq \frac{1 + r - \sigma - (\sigma - 1)}{2} = 1 + \frac{r}{2} - \sigma$}.
\end{proof}

\begin{lemma} \label{lem:sparse-exp-well-defined}
For any $r \in [0, 2]$, there is a unique solution to the equation $\omega(\sigma - 1, 2 - \sigma, 1 + r - \sigma) = \sigma$.
\end{lemma}
\begin{proof}
Fix $r \in [0, 2]$. We start with an observation: By the trivial bounds $a + b \leq \omega(a, b, c) \leq a + b + c$, we have the following lower and upper bounds:
\begin{equation*}
    1 = \sigma - 1 + 2 - \sigma \leq \omega(\sigma - 1, 2 - \sigma, 1 + r - \sigma) \leq \sigma - 1 + 2 - \sigma + 1 + r - \sigma = 2 + r - \sigma.
\end{equation*}
It follows that any feasible solution to the equation $\omega(\sigma - 1, 2 - \sigma, 1 + r - \sigma) = \sigma$ must satisfy $1 \leq \sigma \leq 2 + r - \sigma$ and thus lies in the range $[1, 1 + \frac{r}{2}]$. Next, by the previous lemma the function $g(\sigma) = \sigma - \omega(\sigma - 1, 2 - \sigma, 1 + r - \sigma)$ is strictly increasing on $[1, 1 + \frac{r}{2}]$. Moreover, since $\omega(\cdot, \cdot, \cdot)$ is known to be a continuous function (\cref{fact:omega-continuous}), it follows that also $g$ is continuous. Finally, using again the trivial bounds $b + c \leq \omega(a, b, c) \leq a + b + c$, we can bound~$g$ at
\begin{gather*}
    g(1) = 1 - \omega(0, 1, r) \leq 1 - (1 + r) = -r \leq 0, \\
\intertext{and}
    g(1 + \tfrac r2) = 1 + \tfrac r2 - \omega(\tfrac r2, 1 - \tfrac r2, \tfrac r2) \geq 1 + \tfrac r2 - (\tfrac r2 + 1 - \tfrac r2 + \tfrac r2) = 0.
\end{gather*}
All in all, we showed that $g$ is continuous and strictly increasing on $[1, 1 + \frac{r}{2}]$ and that $g(1) \leq 0$ and $g(1 + \frac{r}{2}) \geq 0$. It follows that there is some $\sigma \in [1 + \frac{r}{2}]$ with $g(\sigma) = 0$, or equivalently, $\omega(\sigma - 1, 2 - \sigma, 1 + r - \sigma) = \sigma$.
\end{proof}

\paragraph{Algebraic Upper Bounds for \boldmath$\sigma(r)$}
We have established that $\sigma(r)$ is well-defined, but it still remains hard to grasp how $\sigma(r)$ behaves as a function of $r$. For this reason, we will now derive some explicit upper bounds on $\sigma(r)$. Our strategy is as follows: We first prove that $\sigma(r)$ is a convex function of $r$. Then, we evaluate $\sigma(r)$ at certain strategic points and conclude by the convexity that the linear interpolation between these points is an upper bound on $\sigma$.

\begin{lemma}[Convexity] \label{lem:sparse-exp-convexity}
The function $\sigma(r)$ is convex.
\end{lemma}
\begin{proof}
Fix any $r_1, r_2 \in [0, 2]$ and any $\lambda \in [0, 1]$; we show that
\begin{equation*}
    \sigma(\lambda r_1 + (1 - \lambda) r_2) \leq \lambda \sigma(r_1) + (1 - \lambda) \sigma(r_2) =: \sigma^*.
\end{equation*}
We start with the following calculation that uses the subadditivity---or equivalently, convexity---of~$\omega$ (see \cref{fact:omega-subadditive}):
\begin{gather*}
    \omega(\sigma^* - 1, 2 - \sigma^*, 1 + \lambda r_1 + (1 - \lambda) r_2 - \sigma^*) \\
    \qquad= \omega(\lambda (\sigma(r_1) - 1) + (1 - \lambda)(\sigma(r_2) - 1), \\
    \qquad\qquad\quad\lambda (2 - \sigma(r_1)) + (1 - \lambda)(2 - \sigma(r_2)), \\
    \qquad\qquad\quad\lambda (1 + r_1 - \sigma(r_1)) + (1 - \lambda)(1 + r_2 - \sigma(r_2))) \\
    \qquad\leq
    \begin{multlined}[t]
        \lambda \cdot \omega(\sigma(r_1) - 1, 2 - \sigma(r_1), 1 + r_1 - \sigma(r_1)) \\+ (1 - \lambda) \cdot \omega(\sigma(r_2) - 1, 2 - \sigma(r_2), 1 + r_2 - \sigma(r_2))
    \end{multlined} \\
    \qquad= \lambda \sigma(r_1) + (1 - \lambda) \sigma(r_2) \\
    \qquad= \sigma^*.
\end{gather*}
By \cref{lem:sparse-exp-increasing} the function $g(\sigma) = \omega(\sigma - 1, 2 - \sigma, 1 + \lambda r_1 + (1 - \lambda) r_2 - \sigma)$ is strictly increasing. Since we just proved that $g(\sigma^*) \geq 0$, it follows that the unique zero $\sigma(\lambda r_1 + (1 - \lambda) r_2)$ of $g$ satisfies \makebox{$\sigma(\lambda r_1 + (1 - \lambda) r_2) \leq \sigma^*$}.
\end{proof}

\begin{lemma}[Trivial Bounds] \label{lem:sparse-exp-trivial-bounds}
$\max\set{1, r} \leq \sigma(r) \leq 1 + \frac{r}{2}$.
\end{lemma}
\begin{proof}
Recall that $\sigma(r) = \omega(\sigma(r) - 1, 2 - \sigma(r), 1 + r - \sigma(r))$. On the one hand, from the trivial lower bound $\max\set{a + b, a + c} \leq \omega(a, b, c)$ (\cref{fact:omega-trivial-bounds}), we obtain $\sigma(r) \geq \max\set{\sigma(r) - 1 + 2 - \sigma(r), \sigma(r) - 1 + 1 + r - \sigma(r)} = \max\set{1, r}$. On the other hand, the trivial upper bound $\omega(a, b, c) \leq a +b + c$ (\cref{fact:omega-trivial-bounds}) entails that $\sigma(r) \leq \sigma(r) - 1 + 2 - \sigma(r) + 1 + r - \sigma(r)$ which can be rewritten as $\sigma(r) \leq 1 + \frac{r}{2}$.
\end{proof}

Note that these trivial bounds already imply tight values $\sigma(0) = 1$ and $\sigma(2) = 2$. We continue to evaluate $\sigma(r)$ for more points. For the following two lemmas, recall that we define the rectangular matrix multiplication constant~$\mu$ as the unique solution to $\omega(\mu, 1, 1) \leq 1 + 2\mu$, and $\alpha = \max\set{\alpha : \omega(\alpha, 1, 1) = 2}$.

\begin{lemma} \label{lem:sparse-exp-rho}
$\sigma(1) = 1 + \frac{\mu}{1+\mu}$.
\end{lemma}
\begin{proof}
Recall that $\sigma = \sigma(1)$ is the unique solution to the equation (1) $\omega(\sigma - 1, 2 - \sigma, 2 - \sigma) = \sigma$ and that $\mu$ is the unique solution to the equation (2) $\omega(\mu, 1, 1) = 1 + 2\mu$, or equivalently (3) \smash{$\omega(\frac{\mu}{1+\mu}, \frac{1}{1+\mu}, \frac{1}{1+\mu}) = \frac{1+2\mu}{1+\mu}$}. By substituting $\sigma$ by \smash{$1 + \frac{\mu}{1+\mu} = \frac{1 + 2\mu}{1+\mu}$} in (1), we find that (1) and (3) are equivalent.
\end{proof}

\begin{lemma} \label{lem:sparse-exp-alpha}
$\sigma(1 + \frac{1}{1+\alpha}) = 1 + \frac{1}{1+\alpha}$.
\end{lemma}
\begin{proof}
We have to prove that $\omega(1 + \frac{1}{1+\alpha} - 1, 2 - (1 + \frac{1}{1+\alpha}), 2 + \frac{1}{1+\alpha} - (1 + \frac{1}{1+\alpha})) = \omega(\frac{1}{1+\alpha}, \frac{\alpha}{1+\alpha}, 1) = 1 + \frac{1}{1+\alpha}$. On the one hand, from the trivial lower bound $a + c \leq \omega(a, b, c)$ (\cref{fact:omega-trivial-bounds}), we get that \smash{$\omega(\frac{1}{1+\alpha}, \frac{\alpha}{1+\alpha}, 1) \geq 1 + \frac{1}{1+\alpha}$}. On the other hand, by \cref{fact:omega-scale,fact:omega-subadditive} we have that
\begin{equation*}
    \omega\parens*{\frac{1}{1+\alpha}, \frac{\alpha}{1+\alpha}, 1} \leq \omega\parens*{\frac{1}{1+\alpha}, \frac{\alpha}{1+\alpha}, \frac{1}{1+\alpha}} + \frac{\alpha}{1+\alpha} = \frac{2}{1+\alpha} + \frac{1}{1+\alpha} = 1 + \frac{1}{1+\alpha}. \qedhere
\end{equation*}
\end{proof}

In summary, we have the identities \smash{$\sigma(0) = 1,\, \sigma(1) = 1 + \frac{\mu}{1+\mu},\, \sigma\parens{1 + \frac{1}{1+\alpha}} = 1 + \frac{1}{1+\alpha},\, \sigma(2) = 2$} by \cref{lem:sparse-exp-trivial-bounds,lem:sparse-exp-rho,lem:sparse-exp-alpha}. By the convexity of $\sigma(r)$ it follows that $\sigma(r)$ is upper-bounded by the line segments that interpolate between these four points:

\begin{lemma}
$\sigma(r) \leq \max\set{1 + r \cdot \frac{\mu}{1+\mu}, \frac{(2+\alpha)\mu}{1+\mu} + r \cdot \frac{1-\alpha\mu}{1+\mu}, r}$, for any $r \in [0, 2]$.
\end{lemma}

\paragraph{Numerical Upper Bounds for \boldmath$\sigma(r)$}
The previous consideration lead to readable upper bounds on $\sigma(r)$, but it is actually possible to improve these upper bounds using the currently best \emph{numerical} bounds on the complexity of rectangular matrix multiplication~\cite{LeGallU18}. We give our results in \cref{tab:sparse-exp-numerical}.

% !TEX root = ../paper.tex
\begin{table}
\caption{Numerical bounds on the sparse matrix multiplication exponent $\sigma(r)$, based on the dense rectangular matrix multiplication bounds by Le Gall and Urrutia~\cite{LeGallU18}.} \label{tab:sparse-exp-numerical}
\begin{tabular*}{0.17\textwidth}[t]{c@{\extracolsep{\fill}}c}
    \toprule
    $r$ & $\sigma(r)$ \\
    \midrule
    0.00 & 1.0000 \\
    0.05 & 1.0171 \\
    0.10 & 1.0342 \\
    0.15 & 1.0513 \\
    0.20 & 1.0684 \\
    0.25 & 1.0855 \\
    0.30 & 1.1026 \\
    0.35 & 1.1197 \\
    0.40 & 1.1368 \\
    \bottomrule
\end{tabular*}
\hfill
\begin{tabular*}{0.17\textwidth}[t]{c@{\extracolsep{\fill}}c}
    \toprule
    $r$ & $\sigma(r)$ \\
    \midrule
    0.45 & 1.1539 \\
    0.50 & 1.1710 \\
    0.55 & 1.1881 \\
    0.60 & 1.2052 \\
    0.65 & 1.2223 \\
    0.70 & 1.2396 \\
    0.75 & 1.2569 \\
    0.80 & 1.2744 \\
    \bottomrule
\end{tabular*}
\hfill
\begin{tabular*}{0.17\textwidth}[t]{c@{\extracolsep{\fill}}c}
    \toprule
    $r$ & $\sigma(r)$ \\
    \midrule
    0.85 & 1.2921 \\
    0.90 & 1.3099 \\
    0.95 & 1.3277 \\
    1.00 & 1.3458 \\
    1.05 & 1.3665 \\
    1.10 & 1.3875 \\
    1.15 & 1.4086 \\
    1.20 & 1.4299 \\
    \bottomrule
\end{tabular*}
\hfill
\begin{tabular*}{0.17\textwidth}[t]{c@{\extracolsep{\fill}}c}
    \toprule
    $r$ & $\sigma(r)$ \\
    \midrule
    1.25 & 1.4513 \\
    1.30 & 1.4728 \\
    1.35 & 1.4943 \\
    1.40 & 1.5199 \\
    1.45 & 1.5476 \\
    1.50 & 1.5761 \\
    1.55 & 1.6060 \\
    1.60 & 1.6378 \\
    \bottomrule
\end{tabular*}
\hfill
\begin{tabular*}{0.17\textwidth}[t]{c@{\extracolsep{\fill}}c}
    \toprule
    $r$ & $\sigma(r)$ \\
    \midrule
    1.65 & 1.6720 \\
    1.70 & 1.7091 \\
    1.75 & 1.7505 \\
    1.80 & 1.8000 \\
    1.85 & 1.8500 \\
    1.90 & 1.9000 \\
    1.95 & 1.9500 \\
    2.00 & 2.0000 \\
    \bottomrule
\end{tabular*}
\end{table}

To achieve these results, suppose that we know a set of bounds of the form $\set{\omega(a_i, b_i, c_i) \leq \omega_i}_{i=1}^n$. We take~\cite[Table~3]{LeGallU18} as the basis for our bounds in \cref{tab:sparse-exp-numerical}. Then, consider the following linear program in the variables $\sigma, \lambda_1, \dots, \lambda_n$:
\begin{equation*}
    \begin{array}{ll@{\;}c@{\;}l}
        \text{minimize} & \sigma, \\[.3ex]
        \text{subject to} & \textstyle \sum_{i=1}^n \lambda_i \omega_i & = & \sigma, \\[.3ex]
        & \textstyle \sum_{i=1}^n \lambda_i a_i & = & 1 + \sigma, \\[.3ex]
        & \textstyle \sum_{i=1}^n \lambda_i b_i & = & 2 - \sigma, \\[.3ex]
        & \textstyle \sum_{i=1}^n \lambda_i c_i & = & 1 + r - \sigma, \\[.3ex]
        & \sigma, \lambda_1, \dots, \lambda_n \geq 0.\hspace{-10cm}\mbox{}
    \end{array}
\end{equation*}
Using \cref{fact:omega-subadditive}, it is easy to show that any solution $\sigma^*$ to this linear program yields an upper bound of the form $\sigma(r) \leq \sigma^*$.

\paragraph{Bounds when \boldmath$\omega = 2$.}
Finally, suppose that $\omega = 2$. In this case, we have $\omega(a, b, c) = \max\set{a + b, a + c, b + c}$. It follows that $\sigma(r)$ is the unique solution to the equation $\max\set{1, r, 3 + r - 2\sigma} = \sigma$, and the following lemma is immediate:

\begin{lemma}
If $\omega = 2$, then $\sigma(r) = \max\set{1 + \frac r3, r}$.
\end{lemma}
% !TEX root = ../paper.tex
\section{Relation to All-Edges Triangle} \label{sec:lower-bounds}
In this section we prove conditional lower bounds for sparse (Boolean) matrix multiplication under a hypothesis (\cref{hypo:AET}) about the All-Edges Triangle problem ($\AETriangle$, \cref{def:AET}). We also prove that in the \emph{fully sparse} setting (where $\IN + \OUT \leq m$), sparse (Boolean) matrix multiplication is \emph{equivalent} to a certain parameterization of $\AETriangle$. We start with a recap of the relevant definitions.

\defaetriangle*

\begin{definition}[$\#\AETriangle$]
The $\#\AETriangle(x, y, z, m)$ problem is to \emph{count}, in a given tripartite graph $G = (X, Y, Z, E)$ with $|X| \leq x,\, |Y| \leq y,\, |Z| \leq z$ and $|E| \leq m$, for each edge $(i, j) \in (X \times Z) \cap E$ how many triangles it is part of.
\end{definition}

\defpsaetriangle*

\hypoaet*

Recall that this hypothesis morally expresses that the best way to detect triangles (for all edges) in a graph is to combine two algorithms: Fast matrix multiplication, and enumerating all 2-paths in the graph. Our hardness result is that, under this hypothesis, sparse matrix multiplication has exponent $\sigma(r)$:

\thmlowerboundpsaetriangle*
\begin{proof}
Fix $r \in [0, 2]$, let $\sigma = \sigma(r)$, and suppose that sparse Boolean matrix multiplication is in time~\makebox{$\Order(\IN^{\sigma-\epsilon})$}. We prove that the $\PSAETriangle(m^a, m^b, m^c, m)$ problem for $a = \sigma - 1,\, b = 2 - \sigma,\, c = 1 + r - \sigma$ can be solved polynomially faster than $m^{\min\set{1+a,\, \omega(a, b, c)}}$. Let $G = (X, Y, Z, E)$ be a given instance of $\PSAETriangle(m^a, m^b, m^c, m)$. We solve this instance as follows: Let $A$ be the adjacency matrix of the bipartite subgraph with vertices~$X \cup Y$, and let~$B$ be the adjacency matrix of the bipartite subgraph with vertices $Y \cup Z$. We compute the matrix product~$A B$ using the efficient algorithm for sparse Boolean matrix multiplication---note that the input size is at most $x y + m = \Order(m)$ and the output size is at most $x z = m^r$ which yields the correct input-to-output ratio (possibly after padding). We report all edges $(i, j) \in (X \times Z) \cap E$ with $(A B)[i, j]$.

The correctness is easy: For any pair $(i, j) \in X \times Z$ we have $(A B)[i, j] = 1$ if and only if there is a 2-path in $G$ from $i$ to $j$ via some node $k \in Y$. Thus, the algorithm reports exactly all pairs $(i, j) \in X \times Z$ for which there is 2-path $(i, k, j)$ and there is an edge $(i, j) \in E$. This is the set of pairs involved in a triangle.

Next, we analyze the running time. The dominating step is the sparse Boolean matrix multiplication in time $\Order(m^{\sigma - \epsilon})$; afterwards, the reporting step runs in negligible time $\Order(y z) = \Order(m)$. But recall that $\sigma$ satisfies the equation $\omega(\sigma + 1, 2 - \sigma, 1 + r - \sigma) = \sigma$, and thus the running time can be written as
\begin{equation*}
    \Order(m^{\sigma - \epsilon}) = \Order(m^{\min\set{\sigma,\, \omega(\sigma + 1, 2 - \sigma, 1 + r - \sigma)}-\epsilon}) = \Order(m^{\min\set{1+a,\, \omega(a, b, c)} - \epsilon}),
\end{equation*}
which contradicts the $\PSAETriangle$ hypothesis.
\end{proof}

\subsection{Equivalence with All-Edges Triangle in the Fully Sparse Setting}
In this section we prove that in the fully sparse setting (i.e., when we measure the combined input plus output sparsity $m = \IN + \OUT$), the sparse matrix multiplication problem is \emph{equivalent} to a certain parameterization of $\AETriangle$. We start with the following lemmas, based on Le Gall and Urrutia's bounds on rectangular matrix multiplication:

\begin{lemma}[{\cite[Table~3]{LeGallU18}}] \label{lem:fast-rectangular-1.3}
$\omega(1, 1.3, 1) \leq 2.6217$ and $\omega(1, 1.4, 1) \leq 2.7085$.
\end{lemma}

\begin{restatable}{lemma}{lemrhobounds} \label{lem:rho-bounds}
For any $0 \leq \delta \leq \frac{1-\mu}2$, it holds that $\omega(\mu + \delta, 1, 1 - \delta) \leq 1 + 2\mu - 0.02\delta$.
\end{restatable}
\begin{proof}
Let $\gamma = 1 - \frac{\delta}{1-\mu}$ and note that $\delta \leq \frac12 \leq \gamma \leq 1 - \delta$. By the subadditivity of $\omega$ (see \cref{fact:omega-subadditive}), we have that
\begin{gather*}
    \omega(\mu + \delta, 1, 1 - \delta) \\
    \qquad
    \begin{array}{@{}c@{\;\;}c@{\;}lll@{\;}l}
        = & \omega( & 1 - (1 - \mu) \gamma, & 1, & \mu + (1 - \mu) \gamma & ) \\[1ex]
        \leq & \omega( & (\gamma - \delta) \mu, & \gamma - \delta, & \gamma - \delta & ) \\[.1ex]
        + & \omega( & 1 - \gamma - \delta, & 1 - \gamma - \delta, & (1 - \gamma - \delta) \mu & ) \\[.1ex]
        + & \omega( & (1 + \mu) \delta, & 2 \delta, & (1 + \mu) \delta & ).
    \end{array}
\end{gather*}
We can bound these three terms by $(\gamma - \delta)(1 + 2\mu)$, $(1 - \gamma - \delta) (1 + 2 \mu)$ and \smash{$(1+\mu)\delta \cdot \omega(1, \frac{2}{1+\mu}, 1)$}, respectively. Further, we numerically bound \smash{$\omega(1, \tfrac{2}{1+\mu}, 1)$} as follows. Since $\mu \geq \frac12$, we have \smash{$\omega(1, \tfrac{2}{1+\mu}, 1) \leq \omega(1, \frac43, 1)$}. Next, we use that $\omega(\cdot, \cdot, \cdot)$ is convex and can thus be upper-bounded by any linear interpolation between two points. Specifically, we can bound
\begin{equation*}
    \omega(1, \tfrac43, 1) \leq \tfrac23 \cdot \omega(1, 1.3, 1) + \tfrac13 \cdot \omega(1, 1.4, 1) = \tfrac23 \cdot 2.6217 + \tfrac13 \cdot 2.7085 \leq 2.6507
\end{equation*}
by the bounds from the previous \cref{lem:fast-rectangular-1.3}. Therefore:
\begin{gather*}
    \omega(\mu + \delta, 1, 1 - \delta) \\
    \qquad\leq (\gamma - \delta)(1 + 2\mu) + (1 - \gamma - \delta) (1 + 2 \mu) + (1+\mu)\delta \cdot \omega(1, \smash{\tfrac{2}{1+\mu}}, 1) \\
    \qquad\leq 1 + 2\mu - 2\delta (1 + 2\mu) + \delta (1 + \mu) \cdot 2.6507 \\
    \qquad= 1 + 2\mu - \delta (2 + 4\mu - 2.6507 - 2.6507\mu) \\
    \qquad=1 + 2\mu - \delta (1.3493 \mu - 0.6507) \\
    \qquad\leq 1 + 2\mu - \delta (\tfrac12 \cdot 1.3493 - 0.6507) \\
    \qquad\leq 1 + 2\mu - 0.0240 \delta,
\end{gather*}
which completes the proof.
\end{proof}

\thmequivalenceboolean*
\begin{proof}[Proof: (1) implies (2).]
This part is very similar to \cref{thm:lower-bound-ps-ae-triangle}, and due to the similarity we only sketch this part. We use sparse Boolean matrix multiplication to compute all pairs of nodes $(i, j) \in X \times Z$ connected by a 2-path, and return all such pairs that are additionally connected by an edge $(i, j) \in E$. The time complexity is dominated by the Boolean matrix multiplication with input size $n^{1+\mu}$ and output size $n^{1+\mu}$ running in time $\Order((n^{1+\mu})^{1 + \frac{\mu}{1+\mu} - \epsilon}) = \Order(n^{1+2\mu - \epsilon})$.

\medskip\noindent
\emph{(2) implies (1).}
Assume that there is an algorithm $\mathcal A$ for $\AETriangle(n^\mu, n, n, n^{1+\mu})$ in time $\Order(n^{1+2\mu-\epsilon'})$ for some $\epsilon' > 0$. We design an efficient algorithm for sparse Boolean matrix multiplication. Let $A \in \set{0, 1}^{x \times y}$ and $B \in \set{0, 1}^{y \times z}$ be a given instance. By densification (\cref{lem:densification-nonnegative}), we can assume that $x z \leq 8 m$ and by symmetry we assume that $x \leq z$.

Let $\epsilon > 0$ be a parameter to be fixed later. In the same spirit as \cref{lem:input-sparse}, we say that an index~$k \in [y]$ is \emph{light} if \smash{$\abs{\set{i : A[i, k] \neq 0}} \leq m^{\frac{\mu}{1+\mu}-\epsilon}$}, and \emph{heavy} otherwise. We let $y_1$ and $y_2$ denote the number of light and heavy indices, respectively, and subdivide $A$ into submatrices $A_1 \in \set{0, 1}^{x \times y_1}$ and $A_2 \in \set{0, 1}^{x \times y_2}$, where the light indices participate in $A_1$ and the heavy indices participate in $A_2$. We similarly subdivide $B$ into~$B_1 \in \set{0, 1}^{y_1 \times z}$ and $B_2 \in \set{0, 1}^{y_2 \times z}$. Then we run the following two steps:
\begin{enumerate}
\item Compute $C_1 = A_1 \cdot B_1$ exploiting the sparsities of $A_1$ and $B_1$ exactly as in \cref{lem:input-sparse} (by enumerating all 2-paths $(i, k, j) \in [x] \times [y_1] \times [z]$ with $A_1[i, k] = B_1[k, j] = 1$).
\item To compute $C_2 = A_2 \cdot B_2$, we distinguish the following two cases:
\begin{enumerate}[label=2\alph*.]
    \item If $x \geq m^{\frac{\mu}{1+\mu}+300\epsilon}$, then compute $A_2 \cdot B_2$ using fast matrix multiplication (ignoring the assumption that $A_2$ and $B_2$ are $m$-sparse).
    \item If $x < m^{\frac{\mu}{1+\mu}+300\epsilon}$, then we compute $A_2 \cdot B_2$ with the help of algorithm $\mathcal A$. Let $G = (X, Y, Z, E)$ be a tripartite graph with vertex parts $X, Y, Z$ of sizes $|X| = x,\, |Y| = y_2,\,|Z| = z$. Add edges in $X \times Y$ as specified by $A_2$, add edges in $Y \times Z$ as specified by $B_2$ (i.e., view $A_2$ and $B_2$ as the bi-adjacency matrices for these respective parts), and add all edges in $X \times Z$. Run $\mathcal A$ on this instance $G$. We let~$C_2 \in \set{0, 1}^{x \times z}$ be the matrix with $C_2[i, j] = 1$ if and only if the edge $(i, j) \in X \times Z$ participated in a triangle in $G$.
\end{enumerate}
\end{enumerate}
Finally, report $C = C_1 + C_2$ as the output.

\medskip
Due to their similarity to \cref{lem:input-sparse} we omit the correctness proof of steps 1 and 2a, and only analyze step 2b. By the construction of the graph $G$, there is a 2-path between two nodes $i \in X$ and $j \in Z$ if and only if there is some node~\makebox{$k \in Y$} with~$A[i, k] = B[k, j] = 1$. Since the graph contains \emph{all} edges~$(i, j)$ any such 2-path can be completed to a triangle. Therefore, $G$ indeed contains a triangle involving $(i, j)$ if and only if~\makebox{$C_2[i, j] = 1$}.

Next, focus on the running time. As before, step 1 runs in time \smash{$\Order(m \cdot m^{\frac{\mu}{1+\mu}-\epsilon}) = \Order(m^{1+\frac{\mu}{1+\mu}-\epsilon})$}. To bound the running time of step 2, we may assume that $x \geq m^{\frac{\mu}{1+\mu}-\epsilon}$ (as otherwise there is no heavy index $k$ which renders step 2 trivial) and that \smash{$y_2 \leq m / m^{\frac{\mu}{1+\mu}-\epsilon} = m^{\frac{1}{1+\mu}+\epsilon}$}.

We start with step 2a: By the three assumptions that \smash{$x \geq m^{\frac{\mu}{1+\mu}+300\epsilon}$}, that $x z \leq 8m$ and that $x \leq z$, there must be some constant $\delta$ with $300 \epsilon \leq \delta \leq \frac{1-\mu}{2}$ such that \smash{$x \leq m^{\frac{\mu+\delta}{1+\mu}+\order(1)}$} and \smash{$z \leq m^{\frac{1-\delta}{1+\mu}+\order(1)}$}. Therefore, it takes time \smash{$\Order(m^{\omega(\frac{\mu+\delta}{1+\mu}, \frac{1}{1+\mu}+\epsilon, \frac{1-\delta}{1+\mu})+\epsilon})$}, say, to compute~$C_2$ by fast matrix multiplication. Using \cref{lem:rho-bounds,fact:omega-scale,fact:omega-subadditive} this exponent can be bounded by
\begin{gather*}
    \omega\parens*{\frac{\mu+\delta}{1+\mu}, \frac{1}{1+\mu} + \epsilon, \frac{1-\delta}{1+\mu} - \delta} + \epsilon \\
    \qquad\leq \frac{\omega(\mu + \delta, 1, 1 - \delta)}{1+\mu} + 2\epsilon \\
    \qquad\leq \frac{1 + 2\mu - 0.02\delta}{1+\mu} + 2\epsilon \\
    \qquad\leq \frac{1 + 2\mu}{1+\mu} + -0.01 \delta + 2\epsilon \\
    \qquad\leq 1 + \frac{\mu}{1+\mu} - \epsilon.
\end{gather*}

It remains to analyze the running time of step 2b. Recall that we can assume that~\smash{$y_2 \leq m^{\frac{1}{1+\mu}+\epsilon}$} and that~\smash{$x \geq m^{\frac{\mu}{1+\mu}-\epsilon}$} and thus~\smash{$z \leq m^{\frac{1}{1+\mu}+\epsilon}$}. By step 2b we further have \smash{$x < m^{\frac{\mu}{1+\mu}+300\epsilon}$}. Let \smash{$n = m^{\frac{1}{1+\mu}+600\epsilon}$}, then these bounds imply that vertex parts in the graph $G$ have sizes $|X| = x \leq n^\mu,\, |Y| = y_2 \leq n,\, |Z| = z \leq n$. Moreover, the number of edges in the graph $G$ is at most $m + x z \leq \Order(m) = \Order(n^{1+\mu})$. Therefore, the graph $G$ is an instance of $\AETriangle(n^\mu, n, n, \Order(n^{1+\mu}))$ and can be solved by $\mathcal A$ in time
\begin{equation*}
    \Order(n^{1+2\mu-\epsilon'}) = \Order(m^{(\frac{1}{1+\mu}+600\epsilon)(1 + 2\mu - \epsilon')}) = \Order(m^{1+\frac{\mu}{1+\mu}+1800\epsilon - \frac12\epsilon'}).
\end{equation*}
(In the last step we used the trivial bounds $\frac12 \leq \mu \leq 1$.) Setting $\epsilon = \frac{\epsilon'}{3602}$, this becomes \smash{$\Order(m^{1+\frac{\mu}{1+\mu}-\epsilon})$} and also the total running time is \smash{$\Order(m^{1+\frac{\mu}{1+\mu}-\epsilon})$}.
\end{proof}

In fact, this equivalence between Boolean matrix multiplication and $\AETriangle$ can be adapted to an equivalence between \emph{integer} matrix multiplication and $\#\AETriangle$:

\begin{theorem}[Equivalence with $\#\AETriangle$] \label{thm:equivalence-integer}
The following four statements are equivalent in terms of randomized algorithms:
\begin{enumerate}[label=(\arabic*)]
    \item There is some $\epsilon > 0$ such that sparse integer matrix multiplication (with entries bounded by $\poly(m)$) is in time $\Order(m^{1+\frac{\mu}{1+\mu}-\epsilon})$.
    \item There is some $\epsilon > 0$ such that sparse nonnegative integer matrix multiplication (with entries bounded by~$\poly(m)$) is in time $\Order(m^{1+\frac{\mu}{1+\mu}-\epsilon})$.
    \item There is some $\epsilon > 0$ such that sparse integer matrix multiplication of $\set{0, 1}$-matrices is in time $\Order(m^{1+\frac{\mu}{1+\mu}-\epsilon})$.
    \item There is some $\epsilon' > 0$ such that $\#\textsc{AllEdgesTriangle}(n^\mu, n, n, n^{1+\mu})$ is in time $\Order(n^{1+2\mu-\epsilon'})$.
\end{enumerate}
\end{theorem}
\begin{proof}[Proof: (1) implies (4).]
This part of the proof is again very similar to \cref{lem:input-sparse}. The only difference is that since integer matrix multiplication supports to \emph{count} the number of 2-paths between two nodes $i$ and~$j$, we can also count the number of triangles involving $(i, j)$.

\medskip\noindent
\emph{(4) implies (3).}
This part of the proof is similar to \cref{thm:equivalence-boolean}. The steps 1 and 2a are already computing the integer-valued matrices $C_1$ and $C_2$. Since we restrict the matrices to have entries $\set{0, 1}$, in step 2b we can exactly express the matrix multiplication problem as an instance of $\#\AETriangle$. We omit further details and instead focus on the new aspects of this proof.

\medskip\noindent
\emph{(3) implies (2).}
Assume that for some $\epsilon > 0$ there is an \smash{$\Order(m^{1+\frac{\mu}{1+\mu}-\epsilon})$}-time algorithm for sparse matrix multiplication of $\set{0, 1}$-matrices. We give an algorithm to efficiently multiply two \emph{nonnegative} matrices~$A, B$ with entries bounded by $m^c$ for some constant $c$. Let $L = \ceil{c \log (m)}$, and construct the $\set{0, 1}$-matrices~$A_0, \dots, A_L$, where~$A_\ell[i, j] = 1$ if and only if the $\ell$-th bit of $A[i, j]$ is one. We similarly construct $B_0, \dots, B_L$. Note that
\begin{equation*}
    A = \sum_{\ell=0}^L 2^\ell A_\ell, \qquad B = \sum_{\ell=0}^L 2^\ell B_\ell,
\end{equation*}
and thus
\begin{equation*}
    A B = \parens*{\sum_{\ell=0}^L 2^\ell A_\ell} \parens*{\sum_{\ell=0}^L 2^\ell B_\ell} = \sum_{\ell_1 = 0}^L \sum_{\ell_2 = 0}^L 2^{\ell_1 + \ell_2} A_{\ell_1} B_{\ell_2}.
\end{equation*}
Note that the sparsity of $A_{\ell_1}$, $B_{\ell_2}$ and $A_{\ell_1} B_{\ell_2}$ does not blow up---more precisely, $\supp(A_{\ell_1}) \subseteq \supp(A)$, $\supp(B_{\ell_2}) \subseteq \supp(B)$ and $\supp(A_{\ell_1} B_{\ell_2}) \subseteq \supp(A B)$ for all $\ell_1, \ell_2$. Hence, we can compute the $L^2 = \Order(\log^2 m)$ matrix products $A_{\ell_1} B_{\ell_2}$ in time \smash{$\widetilde\Order(m^{1+\frac{\mu}{1+\mu}-\epsilon})$}, and obtain $AB$ by the previous equation.

\medskip\noindent
\emph{(2) implies (1).}
Assume that for some $\epsilon > 0$ there is an $\Order(m^{1+\frac{\mu}{1+\mu}-\epsilon})$-time algorithm for sparse nonnegative matrix multiplication. We give an algorithm to efficiently multiply two \emph{integer} matrices~\makebox{$A \in \Int^{x \times y}$} and~\makebox{$B \in \Int^{y \times z}$} (with possibly negative entries). By the randomized densification from \cref{lem:densification-integer}, we may assume that~\makebox{$x z \leq m^{1+\order(1)}$} at the cost of worsening the running time by a polylogarithmic factor.

Let $\Delta$ denote the largest entry in $A$ and $B$ in absolute value. We define two nonnegative matrices~$A_0, A_1$ as follows:
\begin{equation*}
    A_0[i, j] =
    \begin{cases}
        A[i, j] + \Delta &\text{if $A[i, j] \neq 0$,} \\
        0 &\text{otherwise,}
    \end{cases}
    \qquad A_1[i, j] =
    \begin{cases}
        \Delta &\text{if $A[i, j] \neq 0$,} \\
        0 &\text{otherwise.}
    \end{cases}
\end{equation*}
Note that $A = A_0 - A_1$. For similarly defined nonnegative matrices $B_0, B_1$ we have $B = B_0 - B_1$. It follows that
\begin{equation*}
    A B = (A_0 - A_1)(B_0 - B_1) = A_0 B_0 - A_0 B_1 - A_1 B_0 + A_1 B_1.
\end{equation*}
Since $\supp(A_0), \supp(A_1) \subseteq \supp(A)$ and similarly $\supp(B_0), \supp(B_1) \subseteq \supp(B)$, the number of nonzero entries in $A_0, A_1, B_0, B_1$ is bounded by $m$. Additionally, since all four products have dimensions $x \times z$ and we assumed that $x z \leq m^{1+\order(1)}$, the output size is trivially bounded by $m^{1+\order(1)}$. Therefore, we can compute the four matrix products in time $m^{1+\frac{\mu}{1+\mu}-\epsilon+\order(1)} \leq \Order(m^{1+\frac{\mu}{1+\mu}-\epsilon/2})$, and compute $AB$ by the previous equation.
\end{proof}
\pagebreak[2]
% !TEX root = ../paper.tex
\section{Conclusions and Open Problems} \label{sec:conclusions}
We conclude with two important open questions.

\paragraph{Almost-Linear Time?}
If $\omega=2$ then matrix multiplication can be solved in linear time in the \emph{dense} case, but the $m^{4/3}$ barrier persists in the fully sparse case $m =\IN+\OUT$. The non-existence of linear-time algorithms for sparse matrix multiplication has been used as hardness assumption in~\cite{BerkholzGS20}. \cref{thm:equivalence-boolean} shows that future research can focus on a concrete special case, namely the $\AETriangle(\sqrt n, n, n, n^{3/2})$ problem. Any new techniques either for algorithms or for reductions (from other famous problems) should be tested against it.

\paragraph{Derandomization.}
Can we solve integer matrix multiplication in time $O(m^{1+\frac{\mu}{1+\mu}+\epsilon})$ \emph{deterministically}? A full derandomization of our $O(\IN^{\sigma(r)+\epsilon})$-time algorithm for $\OUT = \Theta(\IN^r)$ for all $r\in [0,2]$ is certainly challenging, as it would imply a deterministic $O(m^{1+\epsilon})$-time algorithm for verifying whether three given matrices $A,B,C$ with $m$ nonzeros satisfy $C=AB$, see \cref{sec:introduction:sec:apps}. Such a derandomization of Freivalds' algorithm, which even applies to sparse matrices, is perhaps too strong to hope for. The next best goal is to obtain a derandomization in the relaxed setting (pursued, e.g., in~\cite{Kutzkov13,Kunnemann18}) in which we are given a close estimate $b$ that satisfies $\OUT \le b \le \Order(\OUT)$.

\bibliographystyle{alphaurl}
\bibliography{references}

\appendix
% !TEX root = ../paper.tex
\section{Yuster-Zwick Is Conditionally Optimal} \label{sec:yuster-zwick}
In this short section we provide some evidence that Yuster and Zwick's algorithm~\cite{YusterZ05} for input-sparse matrix multiplication is optimal. Recall that their algorithm computes the product of two sparse~\makebox{$n \times n$} matrices with at most $m$ nonzeros (and without any bound on the number of nonzeros in the output matrix) in time \smash{$m^{\frac{2\omega-4}{\omega-1-\alpha}} n^{\frac{2-\alpha\omega}{\omega-1-\alpha}+\order(1)}$}. For the current values of $\omega \leq 2.3719$~\cite{DuanWZ23} and $\alpha \geq 0.3139$~\cite{LeGallU18}, this running time becomes $\Order(m^{0.704} n^{1.186})$. Note that if $\omega = 2$, the Yuster-Zwick algorithm becomes irrelevant as we can solve the problem in optimal time $n^{2\pm \order(1)}$. We will therefore assume throughout this section that $\omega > 2$. It is easy to check that the Yuster-Zwick algorithm beats the matrix multiplication time $n^\omega$ in the regime with at most \smash{$m \ll n^{\frac{1+\omega}{2}}$} nonzeros. For exactly \smash{$m = \Theta(n^{\frac{1+\omega}{2}})$} nonzeros, their algorithm recovers the matrix multiplication running time. In this section, based on the previous work of the fine-grained complexity community, we show that improving upon Yuster-Zwick in the regime \smash{$m = \Theta(n^{\frac{1+\omega}{2}})$} would contradict some recent fine-grained assumptions. Consider the following problem:

\begin{definition}[Monochromatic All-Edges Triangle]
The Monochromatic All-Edges Triangle problem is, given an edge-colored graph, to decide for each edge whether it is part of a \emph{monochromatic} triangle (i.e., a triangle in which all three edges have the same color).
\end{definition}

The Monochromatic All-Edges Triangle problem can be solved in time $n^{\frac{3+\omega}{2}+\order(1)}$ (this is typically called an \emph{intermediate} running time---between fast matrix multiplication and cubic-time brute-force), and it is a recent conjecture that this time is optimal (up to subpolynomial factors)~\cite{LincolnPW20,VassilevskaWY06,VassilevskaWilliamsX20}:

\begin{hypothesis}[Monochromatic All-Edges Triangle]
The Monochromatic All-Edges Triangle problem cannot be solved in time \smash{$\Order(n^{\frac{3+\omega}{2}-\epsilon})$}, for any $\epsilon > 0$.
\end{hypothesis}

Evidence for this hypothesis is that any improvement for Monochromatic All-Edges Triangle beyond the $n^{\frac{3+\omega}{2}}$ barrier carries over to other well-studied intermediate problems, including the $(\min, \max)$-Product, $\exists$Dominance Product, $\exists$Equality Product, and many more~\cite{VassilevskaWilliamsX20} (assuming that $\omega > 2$). While this hypothesis is much more recent than many other well-established conjectures in fine-grained complexity (and may therefore seem less believable), it can certainly be viewed as an important algorithmic barrier that needs to be overcome to make progress on several interesting problems.

Based on similar reductions as in~\cite{LincolnPW20}, we prove that the Monochromatic All-Edges Triangle hypothesis implies that Yuster and Zwick's algorithm is optimal. The proof is simple, but we are not aware of any prior references.

\begin{lemma}
Assume that $\omega > 2$. Then the Boolean matrix product of two $n \times n$ matrices with at most \smash{$\Order(n^{\frac{1+\omega}{2}})$} nonzero entries cannot be computed in time $\Order(n^{\omega-\epsilon})$, for any $\epsilon > 0$, unless the Monochromatic All-Edges Triangle Hypothesis fails.
\end{lemma}
\begin{proof}
We design a reduction from the Monochromatic All-Edges Triangle problem. We will treat all colors separately. So fix any color $\chi$, let $G_\chi$ denote the subgraph with edges colored with $\chi$, and let~$m_\chi$ denote the number of edges in $G_\chi$. We distinguish the following three cases, based on the frequency~$m_\chi$. Throughout, let $\delta > 0$ be a parameter to be fixed later.
\begin{itemize}
    \item If $n^{\frac{1+\omega}{2}+\delta} \leq m_\chi$: Solve the All-Edges Triangle problem on $G_\chi$ in time $n^{\omega+\order(1)}$ using fast (Boolean) matrix multiplication.
    \item If $n^{\frac{1+\omega}{2}-\delta} \leq m_\chi \leq n^{\frac{1+\omega}{2}+\delta}$: By adding $n' = n^{1+\delta}$ isolated dummy nodes to the graph $G_\chi$, we obtain a larger graph $G_\chi'$ with $\Order(n')$ nodes and
    \begin{equation*}
        m_\chi \leq n^{\frac{1+\omega}{2}+\delta} \leq n^{(1+\delta) \cdot \frac{1+\omega}{2}} \leq \Order((n')^{\frac{1+\omega}{2}})
    \end{equation*}
    edges. We can use the oracle for input-sparse matrix multiplication to solve the All-Edges Triangle problem on that graph in time $\Order((n')^{\omega-\epsilon}) \leq \Order(n^{\omega+\delta\omega-\epsilon}) \leq \Order(n^{\omega+3\delta-\epsilon})$.
    \item If $m_\chi \leq n^{\frac{1+\omega}{2}-\delta}$: Solve the All-Edges Triangle problem on $G_{\chi}$ in time \smash{$(m_\chi)^{\frac{2\omega}{\omega+1}+\order(1)}$}~\cite{AlonYZ97}.
\end{itemize}
This completes the description of the algorithm. It remains to bound the running time. Clearly there are at most \smash{$n^2 / n^{\frac{1+\omega}{2}+\delta} = n^{\frac{3-\omega}{2}-\delta}$} many colors falling into the first category. Solving each such color in time $n^{\omega+\order(1)}$ takes time $n^{\frac{\omega+3}{2}-\delta+\order(1)}$ in total. Similarly, there can be at most $n^{\frac{3-\omega}{2}+\delta}$ colors falling into the second category. For each such color we spend time $\Order(n^{\omega + 3\delta - \epsilon})$, and thus the total time for the second case is \smash{$\Order(n^{\frac{\omega+3}{2}+4\delta-\epsilon})$}. Finally, discarding all colors from the first two categories, the remaining colors satisfy that~\smash{$m_\chi \leq n^{\frac{1+\omega}{2}-\delta}$} and that $\sum_\chi m_\chi \leq n^2$. Thus, the time for the third case is at most
\begin{equation*}
    \sum_\chi (m_\chi)^{\frac{2\omega}{\omega+1}+\order(1)} = \sum_\chi m_\chi \cdot (m_\chi)^{\frac{\omega-1}{\omega+1}+\order(1)} \leq n^2 \cdot (n^{\frac{1+\omega}{2}-\delta})^{\frac{\omega-1}{\omega+1}+\order(1)} \leq n^{\frac{3+\omega}{2} - \frac{\delta}{4}+\order(1)}.
\end{equation*}
All in all, the algorithm runs in time $n^{\frac{3+\omega}{2} - \min(\delta, \epsilon - 4\delta, \frac\delta4)+\order(1)}$. The claim follows for $0 < \delta < \frac{\epsilon}{4}$.
\end{proof}

\end{document}